\documentclass[submission,copyright,creativecommons,sharealike,noncommercial]{eptcs}
\providecommand{\event}{QPL 2018}
\pdfoutput=1

\usepackage{mathtools} 
\usepackage{amssymb} 
\usepackage{bbold} 
\usepackage{amsthm} 
\usepackage{stmaryrd} 

\usepackage{relsize} 
\usepackage{microtype} 
\usepackage{multicol} 
\usepackage{csquotes} 

\usepackage{hyperref} 
\usepackage[nocompress]{cite} 

\usepackage{graphicx} 
\usepackage[usenames,dvipsnames]{xcolor} 
\usepackage{tikz} 
\usepackage{circuitikz} 
\usetikzlibrary{
	arrows,
	shapes,
	decorations,
	intersections,
	backgrounds,
	positioning,
	circuits.ee.IEC
	}

		\newcounter{theorem_c} 
		\numberwithin{theorem_c}{section} 
		\numberwithin{equation}{section} 

		\theoremstyle{plain} 
		\newtheorem{theorem}[theorem_c]{Theorem}
  \newtheorem{proposition}[theorem_c]{Proposition}

		\newtheorem{conjecture}[theorem_c]{Conjecture}
		
		\newtheoremstyle{exampstyle}
		  {2mm} 
		  {2mm} 
		  {\itshape} 
		  {} 
		  {\bfseries} 
		  {.} 
		  {.5em} 
		  {} 

		\theoremstyle{exampstyle}
		\newtheorem{definition}[theorem_c]{Definition}

	\newcommand{\integers}{\mathbb{Z}} 
	\newcommand{\complexs}{\mathbb{C}} 
	
		\newcommand{\ket}[1]{\vert #1 \rangle} 
		\newcommand{\bra}[1]{\langle #1 \vert} 
		\newcommand{\braket}[2]{\langle #1 \vert #2 \rangle} 
		
		\newcommand{\isom}{\cong} 
		\newcommand{\id}[1]{id_{#1}} 

		\newcommand{\tensorUnit}{I} 

		\newcommand{\Hom}[3]{\operatorname{Hom}_{\,#1}\left[#2,#3\right]} 

		\newcommand{\SetCategory}{\operatorname{Set}} 

		\newcommand{\CMonCategory}{\operatorname{CMon}} 
		\newcommand{\RMatCategory}[1]{#1\operatorname{-Mat}} 
		\newcommand{\RModCategory}[1]{#1\operatorname{-Mod}} 
		\newcommand{\fHilbCategory}{\operatorname{fHilb}} 
		\newcommand{\RelCategory}{\operatorname{Rel}} 
		\newcommand{\fRelCategory}{\operatorname{fRel}} 
		\newcommand{\fSetCategory}{\operatorname{fSet}} 

		\newcommand{\CategoryC}{\mathcal{C}}

		\newcommand{\obj}[1]{\operatorname{obj} \, #1} 

	\newcommand{\Xbwcolour}{black!80}

	\newcommand{\Zbwcolour}{white}

	\newcommand{\Ybwcolour}{black!15}

	\newcommand{\Wbwcolour}{black!50}

	\tikzset{
	  rectangle with rounded corners north west/.initial=4pt,
	  rectangle with rounded corners south west/.initial=4pt,
	  rectangle with rounded corners north east/.initial=4pt,
	  rectangle with rounded corners south east/.initial=4pt,
	}
	\makeatletter
	\pgfdeclareshape{rectangle with rounded corners}{
	  \inheritsavedanchors[from=rectangle] 
	  \inheritanchorborder[from=rectangle]
	  \inheritanchor[from=rectangle]{center}
	  \inheritanchor[from=rectangle]{north}
	  \inheritanchor[from=rectangle]{south}
	  \inheritanchor[from=rectangle]{west}
	  \inheritanchor[from=rectangle]{east}
	  \inheritanchor[from=rectangle]{north east}
	  \inheritanchor[from=rectangle]{south east}
	  \inheritanchor[from=rectangle]{north west}
	  \inheritanchor[from=rectangle]{south west}
	  \backgroundpath{
	    \southwest \pgf@xa=\pgf@x \pgf@ya=\pgf@y
	    \northeast \pgf@xb=\pgf@x \pgf@yb=\pgf@y
	    \pgfkeysgetvalue{/tikz/rectangle with rounded corners north west}{\pgf@rectc}
	    \pgfsetcornersarced{\pgfpoint{\pgf@rectc}{\pgf@rectc}}
	    \pgfpathmoveto{\pgfpoint{\pgf@xa}{\pgf@ya}}
	    \pgfpathlineto{\pgfpoint{\pgf@xa}{\pgf@yb}}
	    \pgfkeysgetvalue{/tikz/rectangle with rounded corners north east}{\pgf@rectc}
	    \pgfsetcornersarced{\pgfpoint{\pgf@rectc}{\pgf@rectc}}
	    \pgfpathlineto{\pgfpoint{\pgf@xb}{\pgf@yb}}
	    \pgfkeysgetvalue{/tikz/rectangle with rounded corners south east}{\pgf@rectc}
	    \pgfsetcornersarced{\pgfpoint{\pgf@rectc}{\pgf@rectc}}
	    \pgfpathlineto{\pgfpoint{\pgf@xb}{\pgf@ya}}
	    \pgfkeysgetvalue{/tikz/rectangle with rounded corners south west}{\pgf@rectc}
	    \pgfsetcornersarced{\pgfpoint{\pgf@rectc}{\pgf@rectc}}
	    \pgfpathclose
	 }
	}
	\makeatother

	\tikzset{->-/.style={decoration={markings,mark=at position #1 with {\arrow{>}}},postaction={decorate}}}
	\tikzset{-<-/.style={decoration={markings,mark=at position #1 with {\arrow{<}}},postaction={decorate}}}

	\tikzstyle{every picture}=[baseline=-0.25em,scale=0.5]
	\pgfdeclarelayer{edgelayer}
	\pgfdeclarelayer{nodelayer}
	\pgfsetlayers{edgelayer,nodelayer,main}

	\tikzstyle{box} = [draw,shape=rectangle,inner sep=2pt,minimum height=6mm,minimum width=6mm,fill=white] 
	\tikzstyle{boxlarge} = [draw,shape=rectangle,inner sep=2pt,minimum height=1.5cm,minimum width=8mm,fill=white] 
	\tikzstyle{boxLarge} = [draw,shape=rectangle,inner sep=2pt,minimum height=2cm,minimum width=10mm,fill=white] 
	\tikzstyle{boxsmall} = [draw,shape=rectangle,inner sep=2pt,minimum height=3mm,minimum width=3mm,fill=white] 
	\tikzstyle{dot} = [inner sep=0mm,minimum width=3mm,minimum height=3mm,draw,shape=circle,text depth=-0.1mm]
	\tikzstyle{Zbwdot} = [dot, fill=\Zbwcolour]
	\tikzstyle{Xbwdot} = [dot, fill=\Xbwcolour]
	\tikzstyle{Ybwdot} = [dot, fill=\Ybwcolour]
	\tikzstyle{Wbwdot} = [dot, fill=\Wbwcolour]
	\tikzstyle{antipode} = [boxsmall] 

	\tikzstyle{state} = [draw, rectangle with rounded corners,
	  rectangle with rounded corners north west=8pt,
	  rectangle with rounded corners south west=8pt,
	  rectangle with rounded corners north east=0pt,
	  rectangle with rounded corners south east=0pt,
	,inner sep=2pt,minimum height=6mm,minimum width=6mm,fill=white]
	\tikzstyle{statelarge} = [draw, rectangle with rounded corners,
	  rectangle with rounded corners north west=8pt,
	  rectangle with rounded corners south west=8pt,
	  rectangle with rounded corners north east=0pt,
	  rectangle with rounded corners south east=0pt,
	,inner sep=2pt,minimum height=1.5cm,minimum width=8mm,fill=white]
	\tikzstyle{stateLarge} = [draw, rectangle with rounded corners,
	  rectangle with rounded corners north west=8pt,
	  rectangle with rounded corners south west=8pt,
	  rectangle with rounded corners north east=0pt,
	  rectangle with rounded corners south east=0pt,
	,inner sep=2pt,minimum height=2cm,minimum width=8mm,fill=white]
	\tikzstyle{effect} = [draw, rectangle with rounded corners,
	  rectangle with rounded corners north west=0pt,
	  rectangle with rounded corners south west=0pt,
	  rectangle with rounded corners north east=8pt,
	  rectangle with rounded corners south east=8pt,
	,inner sep=2pt,minimum height=6mm,minimum width=6mm,fill=white]
	\tikzstyle{scalar}=[diamond,draw,inner sep=1pt,font=\small,fill=white]

	\tikzstyle{cdnode}=[fill=white]
	\tikzstyle{labelnode}=[fill=white]
	\tikzstyle{tightlabelnode}=[fill=white,inner sep = 0.1mm]
	\tikzstyle{none}=[inner sep=0pt]
	\tikzstyle{whiteline}=[-, line width=4pt, draw=white]

	\tikzstyle{trace}=[circuit ee IEC,thick,ground,scale=2.5]
	\tikzstyle{cotrace}=[circuit ee IEC,thick,ground,rotate=180,scale=2.5]
	\tikzstyle{upground}=[circuit ee IEC,thick,ground,rotate=90,scale=2.5]
	\tikzstyle{downground}=[circuit ee IEC,thick,ground,rotate=-90,scale=2.5]

	\tikzstyle{doubled} = [line width=1.8pt] 
	
	\tikzstyle{empty diagram}=[draw=gray!40!white,dashed,shape=rectangle,minimum width=1cm,minimum height=1cm]

\usepackage{rotating}
\usepackage{floatpag}
\usepackage{xifthen}
\usepackage[normalem]{ulem}

\usepackage{bm}

\ifthenelse{\isundefined{\showoptional}}{
  \newcommand{\showoptional}{1}
}{}

\ifthenelse{\isundefined{\ismain}}{
  \newcommand{\ismain}{0}
}{}

\usepackage{hyperref}
\usepackage{amsmath,amsthm,amssymb}
\usepackage{wasysym}

\usepackage{xspace,enumerate,color,epsfig}
\usepackage{graphicx}
\graphicspath{{.}{./figures/}}

\usepackage{tikzfig}
\usepackage{stmaryrd}
\usepackage{docmute}
\usepackage{keycommand}

\tikzstyle{env}=[copoint,regular polygon rotate=0,minimum width=0.2cm, fill=black]

\tikzstyle{probs}=[shape=semicircle,fill=white,draw=black,shape border rotate=180,minimum width=1.2cm]

\tikzstyle{every picture}=[baseline=-0.25em,scale=0.5]
\tikzstyle{dotpic}=[] 
\tikzstyle{diredges}=[every to/.style={diredge}]
\tikzstyle{math matrix}=[matrix of math nodes,left delimiter=(,right delimiter=),inner sep=2pt,column sep=1em,row sep=0.5em,nodes={inner sep=0pt},text height=1.5ex, text depth=0.25ex]

\tikzstyle{inline text}=[text height=1.5ex, text depth=0.25ex,yshift=0.5mm]
\tikzstyle{label}=[font=\footnotesize,text height=1.5ex, text depth=0.25ex,yshift=0.5mm]
\tikzstyle{left label}=[label,anchor=east,xshift=1.5mm]
\tikzstyle{right label}=[label,anchor=west,xshift=-1.5mm]

\tikzstyle{braceedge}=[decorate,decoration={brace,amplitude=2mm,raise=-1mm}]
\tikzstyle{small braceedge}=[decorate,decoration={brace,amplitude=1mm,raise=-1mm}]

\tikzstyle{doubled}=[line width=1.6pt] 
\tikzstyle{boldedge}=[doubled,shorten <=-0.17mm,shorten >=-0.17mm]
\tikzstyle{boldedgegray}=[doubled,gray,shorten <=-0.17mm,shorten >=-0.17mm]
\tikzstyle{singleedgegray}=[gray]

\tikzstyle{semidoubled}=[line width=1.4pt] 
\tikzstyle{semiboldedgegray}=[semidoubled,gray,shorten <=-0.17mm,shorten >=-0.17mm]

\tikzstyle{boxedge}=[semiboldedgegray]

\tikzstyle{boldedgedashed}=[very thick,dashed,shorten <=-0.17mm,shorten >=-0.17mm]
\tikzstyle{vboldedgedashed}=[doubled,dashed,shorten <=-0.17mm,shorten >=-0.17mm]
\tikzstyle{left hook arrow}=[left hook-latex]
\tikzstyle{right hook arrow}=[right hook-latex]
\tikzstyle{sembracket}=[line width=0.5pt,shorten <=-0.07mm,shorten >=-0.07mm]

\tikzstyle{causal edge}=[->,thick,gray]
\tikzstyle{causal nondir}=[thick,gray]
\tikzstyle{timeline}=[thick,gray, dashed]

\tikzstyle{cedge}=[<->,thick,gray!70!white]

\tikzstyle{empty diagram}=[draw=gray!40!white,dashed,shape=rectangle,minimum width=1cm,minimum height=1cm]
\tikzstyle{empty diagram small}=[draw=gray!50!white,dashed,shape=rectangle,minimum width=0.6cm,minimum height=0.5cm]

\tikzstyle{dot}=[inner sep=0mm,minimum width=2mm,minimum height=2mm,draw,shape=circle]  
\tikzstyle{Wsquare}=[white dot, shape=regular polygon, rounded corners=0.8 mm, minimum size=3.3 mm, regular polygon sides=3, outer sep=-0.2mm]
\tikzstyle{Wsquareadj}=[white dot, shape=regular polygon, rounded corners=0.8 mm, minimum size=3.3 mm, regular polygon sides=3, outer sep=-0.2mm, regular polygon rotate=180]
\tikzstyle{ddot}=[inner sep=0mm, doubled, minimum width=2.5mm,minimum height=2.5mm,draw,shape=circle]

\tikzstyle{black dot}=[dot,fill=black]
\tikzstyle{white dot}=[dot,fill=white,,text depth=-0.2mm]
\tikzstyle{white Wsquare}=[Wsquare,fill=gray,,text depth=-0.2mm]
\tikzstyle{white Wsquareadj}=[Wsquareadj,fill=white,,text depth=-0.2mm]
\tikzstyle{green dot}=[white dot] 
\tikzstyle{gray dot}=[dot,fill=gray!40!white,,text depth=-0.2mm]
\tikzstyle{red dot}=[gray dot] 

\tikzstyle{black ddot}=[ddot,fill=black]
\tikzstyle{white ddot}=[ddot,fill=white]
\tikzstyle{gray ddot}=[ddot,fill=gray!40!white]

\tikzstyle{gray edge}=[gray!60!white]

\tikzstyle{small dot}=[inner sep=0.5mm,minimum width=0pt,minimum height=0pt,draw,shape=circle]

\tikzstyle{small black dot}=[small dot,fill=black]
\tikzstyle{small white dot}=[small dot,fill=white]
\tikzstyle{small gray dot}=[small dot,fill=gray!40!white]

\tikzstyle{causal dot}=[inner sep=0.4mm,minimum width=0pt,minimum height=0pt,draw=white,shape=circle,fill=gray!40!white]

\tikzstyle{phase dimensions}=[minimum size=5mm,font=\footnotesize,rectangle,rounded corners=2.5mm,inner sep=0.2mm,outer sep=-2mm]
\tikzstyle{dphase dimensions}=[minimum size=5mm,font=\footnotesize,rectangle,rounded corners=2.5mm,inner sep=0.2mm,outer sep=-2mm]

\tikzstyle{white phase dot}=[dot,fill=white,phase dimensions]
\tikzstyle{white phase ddot}=[ddot,fill=white,dphase dimensions]

\tikzstyle{white rect ddot}=[draw=black,fill=white,doubled,minimum size=5mm,font=\footnotesize,rectangle,rounded corners=2.5mm,inner sep=0.2mm]
\tikzstyle{gray rect ddot}=[draw=black,fill=gray!40!white,doubled,minimum size=6mm,font=\footnotesize,rectangle,rounded corners=3mm]

\tikzstyle{gray phase dot}=[dot,fill=gray!40!white,phase dimensions]
\tikzstyle{gray phase ddot}=[ddot,fill=gray!40!white,dphase dimensions]
\tikzstyle{grey phase dot}=[gray phase dot]
\tikzstyle{grey phase ddot}=[gray phase ddot]

\tikzstyle{small phase dimensions}=[minimum size=4mm,font=\tiny,rectangle,rounded corners=2mm,inner sep=0.2mm,outer sep=-2mm]
\tikzstyle{small dphase dimensions}=[minimum size=4mm,font=\tiny,rectangle,rounded corners=2mm,inner sep=0.2mm,outer sep=-2mm]

\tikzstyle{small gray phase dot}=[dot,fill=gray!40!white,small phase dimensions]
\tikzstyle{small gray phase ddot}=[ddot,fill=gray!40!white,small dphase dimensions]

\tikzstyle{small map}=[draw,shape=rectangle,minimum height=4mm,minimum width=4mm,fill=white]

\tikzstyle{cnot}=[fill=white,shape=circle,inner sep=-1.4pt]

\tikzstyle{asym hadamard}=[fill=white,draw,shape=NEbox,inner sep=0.6mm,font=\footnotesize,minimum height=4mm]
\tikzstyle{asym hadamard conj}=[fill=white,draw,shape=NWbox,inner sep=0.6mm,font=\footnotesize,minimum height=4mm]
\tikzstyle{asym hadamard dag}=[fill=white,draw,shape=SEbox,inner sep=0.6mm,font=\footnotesize,minimum height=4mm]

\tikzstyle{hadamard}=[fill=white,draw,inner sep=0.6mm,font=\footnotesize,minimum height=4mm,minimum width=4mm]
\tikzstyle{small hadamard}=[fill=white,draw,inner sep=0.6mm,minimum height=1.5mm,minimum width=1.5mm]
\tikzstyle{small hadamard rotate}=[small hadamard,rotate=45]
\tikzstyle{dhadamard}=[hadamard,doubled]
\tikzstyle{small dhadamard}=[small hadamard,doubled]
\tikzstyle{small dhadamard rotate}=[small hadamard rotate,doubled]
\tikzstyle{antipode}=[white dot,inner sep=0.3mm,font=\footnotesize]

\tikzstyle{scalar}=[diamond,draw,inner sep=0.5pt,font=\small]
\tikzstyle{dscalar}=[diamond,doubled, draw,inner sep=0.5pt,font=\small]

\tikzstyle{small box}=[rectangle,inline text,fill=white,draw,minimum height=5mm,yshift=-0.5mm,minimum width=5mm,font=\small]
\tikzstyle{small gray box}=[small box,fill=gray!30]
\tikzstyle{medium box}=[rectangle,inline text,fill=white,draw,minimum height=5mm,yshift=-0.5mm,minimum width=10mm,font=\small]
\tikzstyle{square box}=[small box] 
\tikzstyle{medium gray box}=[small box,fill=gray!30]
\tikzstyle{semilarge box}=[rectangle,inline text,fill=white,draw,minimum height=5mm,yshift=-0.5mm,minimum width=12.5mm,font=\small]
\tikzstyle{large box}=[rectangle,inline text,fill=white,draw,minimum height=5mm,yshift=-0.5mm,minimum width=15mm,font=\small]
\tikzstyle{large gray box}=[small box,fill=gray!30]

\tikzstyle{Bayes box}=[rectangle,fill=black,draw, minimum height=3mm, minimum width=3mm]

\tikzstyle{gray square point}=[small box,fill=gray!50]

\tikzstyle{dphase box white}=[dhadamard]
\tikzstyle{dphase box gray}=[dhadamard,fill=gray!50!white]
\tikzstyle{phase box white}=[hadamard]
\tikzstyle{phase box gray}=[hadamard,fill=gray!50!white]

\tikzstyle{point}=[regular polygon,regular polygon sides=3,draw,scale=0.75,inner sep=-0.5pt,minimum width=9mm,fill=white,regular polygon rotate=180]
\tikzstyle{point nosep}=[regular polygon,regular polygon sides=3,draw,scale=0.75,inner sep=-2pt,minimum width=9mm,fill=white,regular polygon rotate=180]
\tikzstyle{copoint}=[regular polygon,regular polygon sides=3,draw,scale=0.75,inner sep=-0.5pt,minimum width=9mm,fill=white]
\tikzstyle{dpoint}=[point,doubled]
\tikzstyle{dcopoint}=[copoint,doubled]

\tikzstyle{pointgrow}=[shape=cornerpoint,kpoint common,scale=0.75,inner sep=3pt]
\tikzstyle{pointgrow dag}=[shape=cornercopoint,kpoint common,scale=0.75,inner sep=3pt]

\tikzstyle{wide copoint}=[fill=white,draw,shape=isosceles triangle,shape border rotate=90,isosceles triangle stretches=true,inner sep=0pt,minimum width=1.5cm,minimum height=6.12mm]
\tikzstyle{wide point}=[fill=white,draw,shape=isosceles triangle,shape border rotate=-90,isosceles triangle stretches=true,inner sep=0pt,minimum width=1.5cm,minimum height=6.12mm,yshift=-0.0mm]
\tikzstyle{wide point plus}=[fill=white,draw,shape=isosceles triangle,shape border rotate=-90,isosceles triangle stretches=true,inner sep=0pt,minimum width=1.74cm,minimum height=7mm,yshift=-0.0mm]

\tikzstyle{wide dpoint}=[fill=white,doubled,draw,shape=isosceles triangle,shape border rotate=-90,isosceles triangle stretches=true,inner sep=0pt,minimum width=1.5cm,minimum height=6.12mm,yshift=-0.0mm]

\tikzstyle{tinypoint}=[regular polygon,regular polygon sides=3,draw,scale=0.55,inner sep=-0.15pt,minimum width=6mm,fill=white,regular polygon rotate=180] 

\tikzstyle{white point}=[point]
\tikzstyle{white dpoint}=[dpoint]
\tikzstyle{green point}=[white point] 
\tikzstyle{white copoint}=[copoint]
\tikzstyle{gray point}=[point,fill=gray!40!white]
\tikzstyle{gray dpoint}=[gray point,doubled]
\tikzstyle{red point}=[gray point] 
\tikzstyle{gray copoint}=[copoint,fill=gray!40!white]
\tikzstyle{gray dcopoint}=[gray copoint,doubled]

\tikzstyle{white point guide}=[regular polygon,regular polygon sides=3,font=\scriptsize,draw,scale=0.65,inner sep=-0.5pt,minimum width=9mm,fill=white,regular polygon rotate=180]

\tikzstyle{black point}=[point,fill=black,font=\color{white}]
\tikzstyle{black copoint}=[copoint,fill=black,font=\color{white}]

\tikzstyle{tiny gray point}=[tinypoint,fill=gray!40!white]

\tikzstyle{diredge}=[->]
\tikzstyle{ddiredge}=[<->]
\tikzstyle{rdiredge}=[<-]
\tikzstyle{thickdiredge}=[->, very thick]
\tikzstyle{pointer edge}=[->,very thick,gray]
\tikzstyle{pointer edge part}=[very thick,gray]
\tikzstyle{dashed edge}=[dashed]
\tikzstyle{thick dashed edge}=[very thick,dashed]
\tikzstyle{thick gray dashed edge}=[thick dashed edge,gray!40]
\tikzstyle{thick map edge}=[very thick,|->]

\makeatletter
\newcommand{\boxshape}[3]{%
\pgfdeclareshape{#1}{
\inheritsavedanchors[from=rectangle] 
\inheritanchorborder[from=rectangle]
\inheritanchor[from=rectangle]{center}
\inheritanchor[from=rectangle]{north}
\inheritanchor[from=rectangle]{south}
\inheritanchor[from=rectangle]{west}
\inheritanchor[from=rectangle]{east}
\backgroundpath{
\southwest \pgf@xa=\pgf@x \pgf@ya=\pgf@y
\northeast \pgf@xb=\pgf@x \pgf@yb=\pgf@y

\@tempdima=#2
\@tempdimb=#3

\pgfpathmoveto{\pgfpoint{\pgf@xa - 5pt + \@tempdima}{\pgf@ya}}
\pgfpathlineto{\pgfpoint{\pgf@xa - 5pt - \@tempdima}{\pgf@yb}}
\pgfpathlineto{\pgfpoint{\pgf@xb + 5pt + \@tempdimb}{\pgf@yb}}
\pgfpathlineto{\pgfpoint{\pgf@xb + 5pt - \@tempdimb}{\pgf@ya}}
\pgfpathlineto{\pgfpoint{\pgf@xa - 5pt + \@tempdima}{\pgf@ya}}
\pgfpathclose
}
}}

\boxshape{NEbox}{0pt}{5pt}
\boxshape{SEbox}{0pt}{-5pt}
\boxshape{NWbox}{5pt}{0pt}
\boxshape{SWbox}{-5pt}{0pt}
\boxshape{EBox}{-3pt}{3pt}
\boxshape{WBox}{3pt}{-3pt}
\makeatother

\tikzstyle{cloud}=[shape=cloud,draw,minimum width=1.5cm,minimum height=1.5cm]

\tikzstyle{map}=[draw,shape=NEbox,inner sep=2pt,minimum height=6mm,fill=white]
\tikzstyle{dashedmap}=[draw,dashed,shape=NEbox,inner sep=2pt,minimum height=6mm,fill=white]
\tikzstyle{mapdag}=[draw,shape=SEbox,inner sep=2pt,minimum height=6mm,fill=white]
\tikzstyle{mapadj}=[draw,shape=SEbox,inner sep=2pt,minimum height=6mm,fill=white]
\tikzstyle{maptrans}=[draw,shape=SWbox,inner sep=2pt,minimum height=6mm,fill=white]
\tikzstyle{mapconj}=[draw,shape=NWbox,inner sep=2pt,minimum height=6mm,fill=white]

\tikzstyle{medium map}=[draw,shape=NEbox,inner sep=2pt,minimum height=6mm,fill=white,minimum width=7mm]
\tikzstyle{medium map dag}=[draw,shape=SEbox,inner sep=2pt,minimum height=6mm,fill=white,minimum width=7mm]
\tikzstyle{medium map adj}=[draw,shape=SEbox,inner sep=2pt,minimum height=6mm,fill=white,minimum width=7mm]
\tikzstyle{medium map trans}=[draw,shape=SWbox,inner sep=2pt,minimum height=6mm,fill=white,minimum width=7mm]
\tikzstyle{medium map conj}=[draw,shape=NWbox,inner sep=2pt,minimum height=6mm,fill=white,minimum width=7mm]
\tikzstyle{semilarge map}=[draw,shape=NEbox,inner sep=2pt,minimum height=6mm,fill=white,minimum width=9.5mm]
\tikzstyle{semilarge map trans}=[draw,shape=SWbox,inner sep=2pt,minimum height=6mm,fill=white,minimum width=9.5mm]
\tikzstyle{semilarge map adj}=[draw,shape=SEbox,inner sep=2pt,minimum height=6mm,fill=white,minimum width=9.5mm]
\tikzstyle{semilarge map dag}=[draw,shape=SEbox,inner sep=2pt,minimum height=6mm,fill=white,minimum width=9.5mm]
\tikzstyle{semilarge map conj}=[draw,shape=NWbox,inner sep=2pt,minimum height=6mm,fill=white,minimum width=9.5mm]
\tikzstyle{large map}=[draw,shape=NEbox,inner sep=2pt,minimum height=6mm,fill=white,minimum width=12mm]
\tikzstyle{large map conj}=[draw,shape=NWbox,inner sep=2pt,minimum height=6mm,fill=white,minimum width=12mm]
\tikzstyle{very large map}=[draw,shape=NEbox,inner sep=2pt,minimum height=6mm,fill=white,minimum width=17mm]

\tikzstyle{medium dmap}=[draw,doubled,shape=NEbox,inner sep=2pt,minimum height=6mm,fill=white,minimum width=7mm]
\tikzstyle{medium dmap dag}=[draw,doubled,shape=SEbox,inner sep=2pt,minimum height=6mm,fill=white,minimum width=7mm]
\tikzstyle{medium dmap adj}=[draw,doubled,shape=SEbox,inner sep=2pt,minimum height=6mm,fill=white,minimum width=7mm]
\tikzstyle{medium dmap trans}=[draw,doubled,shape=SWbox,inner sep=2pt,minimum height=6mm,fill=white,minimum width=7mm]
\tikzstyle{medium dmap conj}=[draw,doubled,shape=NWbox,inner sep=2pt,minimum height=6mm,fill=white,minimum width=7mm]
\tikzstyle{semilarge dmap}=[draw,doubled,shape=NEbox,inner sep=2pt,minimum height=6mm,fill=white,minimum width=9.5mm]
\tikzstyle{semilarge dmap trans}=[draw,doubled,shape=SWbox,inner sep=2pt,minimum height=6mm,fill=white,minimum width=9.5mm]
\tikzstyle{semilarge dmap adj}=[draw,doubled,shape=SEbox,inner sep=2pt,minimum height=6mm,fill=white,minimum width=9.5mm]
\tikzstyle{semilarge dmap dag}=[draw,doubled,shape=SEbox,inner sep=2pt,minimum height=6mm,fill=white,minimum width=9.5mm]
\tikzstyle{semilarge dmap conj}=[draw,doubled,shape=NWbox,inner sep=2pt,minimum height=6mm,fill=white,minimum width=9.5mm]
\tikzstyle{large dmap}=[draw,doubled,shape=NEbox,inner sep=2pt,minimum height=6mm,fill=white,minimum width=12mm]
\tikzstyle{large dmap conj}=[draw,doubled,shape=NWbox,inner sep=2pt,minimum height=6mm,fill=white,minimum width=12mm]
\tikzstyle{large dmap trans}=[draw,doubled,shape=SWbox,inner sep=2pt,minimum height=6mm,fill=white,minimum width=12mm]
\tikzstyle{large dmap adj}=[draw,doubled,shape=SEbox,inner sep=2pt,minimum height=6mm,fill=white,minimum width=12mm]
\tikzstyle{large dmap dag}=[draw,doubled,shape=SEbox,inner sep=2pt,minimum height=6mm,fill=white,minimum width=12mm]
\tikzstyle{very large dmap}=[draw,doubled,shape=NEbox,inner sep=2pt,minimum height=6mm,fill=white,minimum width=19.5mm]

\tikzstyle{muxbox}=[draw,shape=rectangle,minimum height=3mm,minimum width=3mm,fill=white]
\tikzstyle{dmuxbox}=[muxbox,doubled]

\tikzstyle{box}=[draw,shape=rectangle,inner sep=2pt,minimum height=6mm,minimum width=6mm,fill=white]
\tikzstyle{dbox}=[draw,doubled,shape=rectangle,inner sep=2pt,minimum height=6mm,minimum width=6mm,fill=white]
\tikzstyle{dmap}=[draw,doubled,shape=NEbox,inner sep=2pt,minimum height=6mm,fill=white]
\tikzstyle{dmapdag}=[draw,doubled,shape=SEbox,inner sep=2pt,minimum height=6mm,fill=white]
\tikzstyle{dmapadj}=[draw,doubled,shape=SEbox,inner sep=2pt,minimum height=6mm,fill=white]
\tikzstyle{dmaptrans}=[draw,doubled,shape=SWbox,inner sep=2pt,minimum height=6mm,fill=white]
\tikzstyle{dmapconj}=[draw,doubled,shape=NWbox,inner sep=2pt,minimum height=6mm,fill=white]

\tikzstyle{ddmap}=[draw,doubled,dashed,shape=NEbox,inner sep=2pt,minimum height=6mm,fill=white]
\tikzstyle{ddmapdag}=[draw,doubled,dashed,shape=SEbox,inner sep=2pt,minimum height=6mm,fill=white]
\tikzstyle{ddmapadj}=[draw,doubled,dashed,shape=SEbox,inner sep=2pt,minimum height=6mm,fill=white]
\tikzstyle{ddmaptrans}=[draw,doubled,dashed,shape=SWbox,inner sep=2pt,minimum height=6mm,fill=white]
\tikzstyle{ddmapconj}=[draw,doubled,dashed,shape=NWbox,inner sep=2pt,minimum height=6mm,fill=white]

\boxshape{sNEbox}{0pt}{3pt}
\boxshape{sSEbox}{0pt}{-3pt}
\boxshape{sNWbox}{3pt}{0pt}
\boxshape{sSWbox}{-3pt}{0pt}
\tikzstyle{smap}=[draw,shape=sNEbox,fill=white]
\tikzstyle{smapdag}=[draw,shape=sSEbox,fill=white]
\tikzstyle{smapadj}=[draw,shape=sSEbox,fill=white]
\tikzstyle{smaptrans}=[draw,shape=sSWbox,fill=white]
\tikzstyle{smapconj}=[draw,shape=sNWbox,fill=white]

\tikzstyle{dsmap}=[draw,dashed,shape=sNEbox,fill=white]
\tikzstyle{dsmapdag}=[draw,dashed,shape=sSEbox,fill=white]
\tikzstyle{dsmaptrans}=[draw,dashed,shape=sSWbox,fill=white]
\tikzstyle{dsmapconj}=[draw,dashed,shape=sNWbox,fill=white]

\boxshape{mNEbox}{0pt}{10pt}
\boxshape{mSEbox}{0pt}{-10pt}
\boxshape{mNWbox}{10pt}{0pt}
\boxshape{mSWbox}{-10pt}{0pt}
\tikzstyle{mmap}=[draw,shape=mNEbox]
\tikzstyle{mmapdag}=[draw,shape=mSEbox]
\tikzstyle{mmaptrans}=[draw,shape=mSWbox]
\tikzstyle{mmapconj}=[draw,shape=mNWbox]

\tikzstyle{mmapgray}=[draw,fill=gray!40!white,shape=mNEbox]
\tikzstyle{smapgray}=[draw,fill=gray!40!white,shape=sNEbox]

\makeatletter

\pgfdeclareshape{cornerpoint}{
\inheritsavedanchors[from=rectangle] 
\inheritanchorborder[from=rectangle]
\inheritanchor[from=rectangle]{center}
\inheritanchor[from=rectangle]{north}
\inheritanchor[from=rectangle]{south}
\inheritanchor[from=rectangle]{west}
\inheritanchor[from=rectangle]{east}
\backgroundpath{
\southwest \pgf@xa=\pgf@x \pgf@ya=\pgf@y
\northeast \pgf@xb=\pgf@x \pgf@yb=\pgf@y

\pgfmathsetmacro{\pgf@shorten@left}{\pgfkeysvalueof{/tikz/shorten left}}
\pgfmathsetmacro{\pgf@shorten@right}{\pgfkeysvalueof{/tikz/shorten right}}

\pgfpathmoveto{\pgfpoint{0.5 * (\pgf@xa + \pgf@xb)}{\pgf@ya - 5pt}}
\pgfpathlineto{\pgfpoint{\pgf@xa - 8pt + \pgf@shorten@left}{\pgf@yb - 1.5 * \pgf@shorten@left}}
\pgfpathlineto{\pgfpoint{\pgf@xa - 8pt + \pgf@shorten@left}{\pgf@yb}}
\pgfpathlineto{\pgfpoint{\pgf@xb + 8pt - \pgf@shorten@right}{\pgf@yb}}
\pgfpathlineto{\pgfpoint{\pgf@xb + 8pt - \pgf@shorten@right}{\pgf@yb - 1.5 * \pgf@shorten@right}}
\pgfpathclose
}
}

\pgfdeclareshape{cornercopoint}{
\inheritsavedanchors[from=rectangle] 
\inheritanchorborder[from=rectangle]
\inheritanchor[from=rectangle]{center}
\inheritanchor[from=rectangle]{north}
\inheritanchor[from=rectangle]{south}
\inheritanchor[from=rectangle]{west}
\inheritanchor[from=rectangle]{east}
\backgroundpath{
\southwest \pgf@xa=\pgf@x \pgf@ya=\pgf@y
\northeast \pgf@xb=\pgf@x \pgf@yb=\pgf@y

\pgfmathsetmacro{\pgf@shorten@left}{\pgfkeysvalueof{/tikz/shorten left}}
\pgfmathsetmacro{\pgf@shorten@right}{\pgfkeysvalueof{/tikz/shorten right}}

\pgfpathmoveto{\pgfpoint{0.5 * (\pgf@xa + \pgf@xb)}{\pgf@yb + 5pt}}
\pgfpathlineto{\pgfpoint{\pgf@xa - 8pt + \pgf@shorten@left}{\pgf@ya + 1.5 * \pgf@shorten@left}}
\pgfpathlineto{\pgfpoint{\pgf@xa - 8pt + \pgf@shorten@left}{\pgf@ya}}
\pgfpathlineto{\pgfpoint{\pgf@xb + 8pt - \pgf@shorten@right}{\pgf@ya}}
\pgfpathlineto{\pgfpoint{\pgf@xb + 8pt - \pgf@shorten@right}{\pgf@ya + 1.5 * \pgf@shorten@right}}
\pgfpathclose
}
}

\makeatother

\pgfkeyssetvalue{/tikz/shorten left}{0pt}
\pgfkeyssetvalue{/tikz/shorten right}{0pt}

\tikzstyle{kpoint common}=[draw,fill=white,inner sep=1pt,minimum height=4mm]
\tikzstyle{kpoint sc}=[shape=cornerpoint,kpoint common]
\tikzstyle{kpoint adjoint sc}=[shape=cornercopoint,kpoint common]
\tikzstyle{kpoint}=[shape=cornerpoint,shorten left=5pt,kpoint common]
\tikzstyle{kpoint adjoint}=[shape=cornercopoint,shorten left=5pt,kpoint common]
\tikzstyle{kpoint conjugate}=[shape=cornerpoint,shorten right=5pt,kpoint common]
\tikzstyle{kpoint transpose}=[shape=cornercopoint,shorten right=5pt,kpoint common]
\tikzstyle{kpoint symm}=[shape=cornerpoint,shorten left=5pt,shorten right=5pt,kpoint common]

\tikzstyle{wide kpoint sc}=[shape=cornerpoint,kpoint common, minimum width=1 cm]
\tikzstyle{wide kpointdag sc}=[shape=cornercopoint,kpoint common, minimum width=1 cm]

\tikzstyle{black kpoint}=[shape=cornerpoint,shorten left=5pt,kpoint common,fill=black,font=\color{white}]

\tikzstyle{black kpoint sm}=[shape=cornerpoint,shorten left=5pt,kpoint common,fill=black,font=\color{white},scale=0.75]

\tikzstyle{black kpoint adjoint}=[shape=cornercopoint,shorten left=5pt,kpoint common,fill=black,font=\color{white}]
\tikzstyle{black kpointadj}=[shape=cornercopoint,shorten left=5pt,kpoint common,fill=black,font=\color{white}]

\tikzstyle{black kpointadj sm}=[shape=cornercopoint,shorten left=5pt,kpoint common,fill=black,font=\color{white},scale=0.75]

\tikzstyle{black dkpoint}=[shape=cornerpoint,shorten left=5pt,kpoint common,fill=black, doubled,font=\color{white}]
\tikzstyle{black dkpoint adjoint}=[shape=cornercopoint,shorten left=5pt,kpoint common,fill=black, doubled,font=\color{white}]
\tikzstyle{black dkpointadj}=[shape=cornercopoint,shorten left=5pt,kpoint common,fill=black, doubled,font=\color{white}]

\tikzstyle{black dkpoint sm}=[shape=cornerpoint,shorten left=5pt,kpoint common,fill=black, doubled,font=\color{white},scale=0.75]
\tikzstyle{black dkpointadj sm}=[shape=cornercopoint,shorten left=5pt,kpoint common,fill=black, doubled,font=\color{white},scale=0.75] 

\tikzstyle{kpointdag}=[kpoint adjoint]
\tikzstyle{kpointadj}=[kpoint adjoint]
\tikzstyle{kpointconj}=[kpoint conjugate]
\tikzstyle{kpointtrans}=[kpoint transpose]

\tikzstyle{big kpoint}=[kpoint, minimum width=1.2 cm, minimum height=8mm, inner sep=4pt, text depth=3mm]

\tikzstyle{wide kpoint}=[kpoint, minimum width=1 cm, inner sep=2pt]
\tikzstyle{wide kpointdag}=[kpointdag, minimum width=1 cm, inner sep=2pt]
\tikzstyle{wide kpointconj}=[kpointconj, minimum width=1 cm, inner sep=2pt]
\tikzstyle{wide kpointtrans}=[kpointtrans, minimum width=1 cm, inner sep=2pt]

\tikzstyle{wider kpoint}=[kpoint, minimum width=1.25 cm, inner sep=2pt]
\tikzstyle{wider kpointdag}=[kpointdag, minimum width=1.25 cm, inner sep=2pt]
\tikzstyle{wider kpointconj}=[kpointconj, minimum width=1.25 cm, inner sep=2pt]
\tikzstyle{wider kpointtrans}=[kpointtrans, minimum width=1.25 cm, inner sep=2pt]

\tikzstyle{gray kpoint}=[kpoint,fill=gray!50!white]
\tikzstyle{gray kpointdag}=[kpointdag,fill=gray!50!white]
\tikzstyle{gray kpointadj}=[kpointadj,fill=gray!50!white]
\tikzstyle{gray kpointconj}=[kpointconj,fill=gray!50!white]
\tikzstyle{gray kpointtrans}=[kpointtrans,fill=gray!50!white]

\tikzstyle{gray dkpoint}=[kpoint,fill=gray!50!white,doubled]
\tikzstyle{gray dkpointdag}=[kpointdag,fill=gray!50!white,doubled]
\tikzstyle{gray dkpointadj}=[kpointadj,fill=gray!50!white,doubled]
\tikzstyle{gray dkpointconj}=[kpointconj,fill=gray!50!white,doubled]
\tikzstyle{gray dkpointtrans}=[kpointtrans,fill=gray!50!white,doubled]

\tikzstyle{white label}=[draw,fill=white,rectangle,inner sep=0.7 mm]
\tikzstyle{gray label}=[draw,fill=gray!50!white,rectangle,inner sep=0.7 mm]
\tikzstyle{black label}=[draw,fill=black,rectangle,inner sep=0.7 mm]

\tikzstyle{dkpoint}=[kpoint,doubled]
\tikzstyle{wide dkpoint}=[wide kpoint,doubled]
\tikzstyle{dkpointdag}=[kpoint adjoint,doubled]
\tikzstyle{wide dkpointdag}=[wide kpointdag,doubled]
\tikzstyle{dkcopoint}=[kpoint adjoint,doubled]
\tikzstyle{dkpointadj}=[kpoint adjoint,doubled]
\tikzstyle{dkpointconj}=[kpoint conjugate,doubled]
\tikzstyle{dkpointtrans}=[kpoint transpose,doubled]

\tikzstyle{kscalar}=[kpoint common, shape=EBox, inner xsep=-1pt, inner ysep=3pt,font=\small]
\tikzstyle{kscalarconj}=[kpoint common, shape=WBox, inner xsep=-1pt, inner ysep=3pt,font=\small]

\tikzstyle{spekpoint}=[kpoint sc,minimum height=5mm,inner sep=3pt]
\tikzstyle{spekcopoint}=[kpoint adjoint sc,minimum height=5mm,inner sep=3pt]

\tikzstyle{dspekpoint}=[spekpoint,doubled]
\tikzstyle{dspekcopoint}=[spekcopoint,doubled]

 \tikzstyle{upground}=[circuit ee IEC,thick,ground,rotate=90,scale=1.4]
 \tikzstyle{upgroundnormal}=[circuit ee IEC,thick,ground,rotate=90,scale=2]
 \tikzstyle{downground}=[circuit ee IEC,thick,ground,rotate=-90,scale=1.4]
 \tikzstyle{bigground}=[regular polygon,regular polygon sides=3,draw=gray,scale=0.50,inner sep=-0.5pt,minimum width=10mm,fill=gray]

\tikzstyle{arrs}=[-latex,font=\small,auto]
\tikzstyle{arrow plain}=[arrs]
\tikzstyle{arrow dashed}=[dashed,arrs]
\tikzstyle{arrow bold}=[very thick,arrs]
\tikzstyle{arrow hide}=[draw=white!0,-]
\tikzstyle{arrow reverse}=[latex-]
\tikzstyle{cdnode}=[]

\let\olddagger\dagger
\renewcommand{\dagger}{\ensuremath{\olddagger}\xspace}

\usepackage{makeidx}
\makeindex

\newkeycommand{\moral}[width=11cm][1]{\begin{center}
\fbox{\ \parbox{\commandkey{width}}{\centering #1\vphantom{Xy}}\ }
\end{center}}

\newkeycommand{\morallong}[width=11cm][1]{\par\medskip\noindent
\centerline{\fbox{\ \parbox{\commandkey{width}}{\centering #1\vphantom{Xy}}\ }}
\par\medskip\noindent}

\usepackage{color}
\def\bR{\begin{color}{red}}
\def\bB{\begin{color}{blue}}
\def\bM{\begin{color}{magenta}}
\def\bC{\begin{color}{cyan}}
\def\bW{\begin{color}{white}}
\def\bBl{\begin{color}{black}}
\def\bG{\begin{color}{green}}
\def\bY{\begin{color}{yellow}}
\def\jR{\begin{color}{magenta}}
\def\jB{\begin{color}{cyan}}
\def\e{\end{color}\xspace}
\newcommand{\bit}{\begin{itemize}}
\newcommand{\eit}{\end{itemize}\par\noindent}
\newcommand{\ben}{\begin{enumerate}}
\newcommand{\een}{\end{enumerate}\par\noindent}
\newcommand{\beq}{\begin{equation}}
\newcommand{\eeq}{\end{equation}\par\noindent}
\newcommand{\beqa}{\begin{eqnarray*}}
\newcommand{\eeqa}{\end{eqnarray*}\par\noindent}
\newcommand{\beqn}{\begin{eqnarray}}
\newcommand{\eeqn}{\end{eqnarray}\par\noindent}

\def\bR{\begin{color}{red}}  
\def\bB{\begin{color}{blue}} 
\def\bM{\begin{color}{magenta}} 
\def\bGr{\begin{color}{darkgray}}
\def\bC{\begin{color}{cyan}}
\def\bW{\begin{color}{white}}
\def\bBl{\begin{color}{black}}
\def\bG{\begin{color}{green}}
\def\bY{\begin{color}{yellow}}
\def\jR{\begin{color}{magenta}}
\def\jB{\begin{color}{cyan}}
\def\e{\end{color}}

\usepackage{newtxmath}

\title{Symmetric Monoidal Structure \\ with Local Character is a Property}
\author{
	Stefano Gogioso\\
	University of Oxford \\
	\texttt{stefano.gogioso@cs.ox.ac.uk}
	\and
	Dan Marsden\\
	University of Oxford \\
	\texttt{daniel.marsden@cs.ox.ac.uk}
	\and
	Bob Coecke\\
	University of Oxford \\
	\texttt{bob.coecke@cs.ox.ac.uk}
}

\def\titlerunning{Symmetric Monoidal Structure with Local Character is a Property}
\def\authorrunning{S. Gogioso, D. Marsden \& B. Coecke}

\newcommand{\CatUniverse}{\operatorname{Cats}}
\newcommand{\SymMonCatUniverse}{\operatorname{SMCs}}

\newcommand{\RelQCategory}[1]{\operatorname{Rel}(#1)}

\begin{document}

\maketitle

\begin{abstract}
	In previous work we proved that, for categories of free finite-dimensional modules over a commutative semiring, linear compact-closed symmetric monoidal structure is a property, rather than a structure. That is, if there is such a structure, then it is uniquely defined (up to monoidal equivalence). Here we provide a novel unifying category-theoretic notion of symmetric monoidal structure \emph{with local character}, which we prove to be a property for a much broader spectrum of categorical examples, including the infinite-dimensional case of relations over a quantale and the non-free case of finitely generated modules over a principal ideal domain. 
\end{abstract}

\section{Introduction} 
\label{section_introduction}

Is it a property, or is it a structure? That is: Is it enough to state that a mathematical object has a certain feature in order to fully specify that feature, or does one have to provide additional details? A prototypical example of a property is a Cartesian monoidal structure, which arises in an essentially unique way from certain categorical limits---namely products---whenever these exist. An example of a structure,  on the other hand, is a group structure imposed a set: already on a four-element set there are at least two different group structures available.  

This is not merely a question of mathematical interest, but also touches upon the foundations of several scientific domains. In one example, the tensor structure of quantum theory is where the characteristic features of the theory truly emerge. In another example, the tensor structure of certain categories determines the compositional aspects of natural language meaning. As a consequence, the freedom one has in choosing said structure is of fundamental scientific interest.    

When reasoning about physical theories, the kind of tensor structure a theory possesses says something about the nature of the interactions in the theory, or in other words about the behaviour of composite systems. If the tensor is Cartesian, then the state of a joint system can be fully specified by specifying the states of the individual sub-systems. On the opposite end of the spectrum, when the tensor is compact-closed this fails in the most extreme of manners, with the emergence of many states which cannot be understood by looking at sub-systems alone. The diagrammatic language of symmetric monoidal categories makes this very obvious when depicting states of two systems, with Cartesian states always separated and compact-closed states (almost) always connected.
\[
\begin{array}{ccc}
\begin{tikzpicture}
	\begin{pgfonlayer}{nodelayer}
		\node [style=none] (0) at (0.75, 0.75) {};
		\node [style=none] (1) at (-0.75, 0.75) {};
		\node [style=point] (2) at (-0.75, -0.25) {};
		\node [style=point] (3) at (0.75, -0.25) {};
	\end{pgfonlayer}
	\begin{pgfonlayer}{edgelayer}
		\draw [style=swap] (1.center) to (2);
		\draw [style=swap] (0.center) to (3);
	\end{pgfonlayer}
\end{tikzpicture}&\qquad\qquad&\begin{tikzpicture}
	\begin{pgfonlayer}{nodelayer}
		\node [style=none] (0) at (-1, 0.5) {};
		\node [style=none] (1) at (1, 0.5) {}; 
	\end{pgfonlayer}
	\begin{pgfonlayer}{edgelayer}
		\draw [in=-90, out=-90, looseness=1.75] (0.center) to (1.center); 
	\end{pgfonlayer}
\end{tikzpicture}\vspace{2mm}\\
\mbox{Cartesian} &\qquad\qquad& \mbox{compact-closed}
\end{array} 
\]

In our previous work \cite{Uniqueness2017}, we have carried out an initial investigation on the uniqueness of compact-closed symmetric monoidal structure \cite{Kelly,KellyLaplaza}, asking the question whether it be a property for process theories (by which we mean symmetric monoidal categories, or SMCs for short). We have provided a positive answer for categories $\RMatCategory{S}$ of free finite-dimensional modules over a commutative semiring $S$ \cite{FantasticQT2017}. While these categories are of top interest in categorical quantum mechanics \cite{CQM2004,PQP2017} and compositional distributional linguistics \cite{DisCoCat2010}, they constitute a highly restricted family of linear-algebraic categories. Further to its reliance on linear-algebra, our original proof heavily relied on compact-closure, freeness and finite-dimensionality.


In this work, we ditch the model-dependent assumptions of linear-algebra, freeness and finite-dimensionality, as well as the requirement of compact-closure. We replace these by the new category-theoretic notion of \emph{local character} given by some \emph{$\otimes$-free} category, and we show that symmetric monoidal structure with such local character is again essentially unique. Besides the fact that our new framework has a ``pure'' category-theoretic formulation, one major upshot is that 
our uniqueness result now extends to a much broader class of examples, including the infinite-dimensional example of categories of relations over a quantale and the non-free example of categories of modules over a principal ideal domain, which we discuss in Section \ref{section_applications}. 


\section{SMCs with local character} 

We will introduce the new categorical notion of a SMC \emph{with local character}, i.e. one where categorical data specified on it is uniquely determined by data specified on some \emph{$\otimes$-free subcategory}. Here, $\otimes$-free does not just mean that there is no monoidal structure, but also that there is no trace at all of the $\otimes$-structure from the parent SMC, which could, for example, still be present in the factorisation structure of objects. We now proceed to make this intuition formal.

\subsection{\texorpdfstring{$\otimes$-Free Subcategories}{Tensor-Free Subcategories}}


\begin{definition}\em
	Let $\mathcal{C}$ be a SMC, and $\mathcal{A}$ be a sub-category. The \emph{minimal span} $\langle \mathcal{A} \rangle_{\otimes}$ of $\mathcal{A}$ in $\mathcal{C}$ is the smallest sub-SMC of $\mathcal{C}$ which contains $\mathcal{A}$:
	\begin{itemize}
		\item every object in $\langle \mathcal{A} \rangle_{\otimes}$ can be written---up to associators and unitors---as $\otimes_{j=1}^n A_j$ for some family $(A_j)_{j=1}^S$ of objects $A_j \in \obj{\mathcal{A}}$, where the empty tensor product is taken to be the tensor unit;
		\item every morphism in $\langle \mathcal{A} \rangle_{\otimes}$ can be written---up to associators, unitors and symmetry isomorphism---as $\otimes_{j=1}^n f_j$ for some $(f_j: A_j \rightarrow B_j)_{j=1}^n$ is a family of morphisms between objects $A_j,B_j \in \obj{\mathcal{A}}$. 
	\end{itemize}
	The \emph{maximal span} $\overline{\langle\mathcal{A}\rangle}_{\otimes}$ of $\mathcal{A}$ in $\mathcal{C}$ is the smallest full sub-SMC of $\mathcal{C}$ which contains $\mathcal{A}$; equivalently, it is the full sub-SMC of $\mathcal{C}$ spanned by the objects of $\langle\mathcal{A}\rangle_{\otimes}$.
\end{definition}

\begin{definition}\em  
	Let $\mathcal{C}$ be a SMC. If $A,B \in \obj{\mathcal{C}}$ we say that \emph{$A$ $\otimes$-divides $B$}, written $A \vert B$, if there is some $A' \in \obj{\mathcal{C}}$ such that $A'$ is not isomorphic to the tensor unit and $A \otimes A' \isom B$. An object $A \in \obj{\mathcal{C}}$ is said to be \emph{$\otimes$-prime} if: (i) it is not a zero object\footnote{By which we mean one which is absorbing for the tensor product, up to isomorphism.}, (ii) it is not isomorphic to the tensor unit, and (iii) whenever $A | B\otimes C$ we have that $A | B$ or $A | C$. An object $A \in \obj{\mathcal{C}}$ is said to be \emph{uniquely $\otimes$-factorisable} if it is either a zero object or it can be written in a unique way---up to associators, unitors and symmetry isomorphisms---as a tensor product of $\otimes$-prime objects.
\end{definition}
 
\begin{definition}\em
	A SMC $\mathcal{C}$ is \emph{product tomographic} whenever given any two families $(f_j,g_j: A_j \rightarrow B_j)_{j=1}^n$ of processes in $\mathcal{C}$, if for all families of states $(a_j:\tensorUnit \rightarrow A_j)_{j=1}^n$ and effects $(b_j: B_j \rightarrow \tensorUnit)_{j=1}^n$ in $\mathcal{C}$ we have the following equality between scalars:
	\[
		\begin{tikzpicture}
	\begin{pgfonlayer}{nodelayer}
		\node [style=copoint] (0) at (-1.25, 1.25) {$b_1$};
		\node [style=point] (1) at (-1.25, -1.25) {$a_1$};
		\node [style=copoint] (2) at (1.25, 1.25) {$b_n$};
		\node [style=small box] (3) at (-1.25, 0) {$f_1$};
		\node [style=point] (4) at (1.25, -1.25) {$a_n$};
		\node [style=small box] (5) at (1.25, 0) {$f_n$};
		\node [style=none] (6) at (0, 0) {$\cdots$};
	\end{pgfonlayer}
	\begin{pgfonlayer}{edgelayer}
		\draw [style=swap] (3) to (1);
		\draw [style=swap] (5) to (4);
		\draw [style=swap] (0) to (3);
		\draw [style=swap] (2) to (5);
	\end{pgfonlayer}
\end{tikzpicture}\ \ = \ \ \begin{tikzpicture}
	\begin{pgfonlayer}{nodelayer}
		\node [style=copoint] (0) at (-1.25, 1.25) {$b_1$};
		\node [style=point] (1) at (-1.25, -1.25) {$a_1$};
		\node [style=copoint] (2) at (1.25, 1.25) {$b_n$};
		\node [style=small box] (3) at (-1.25, 0) {$g_1$};
		\node [style=point] (4) at (1.25, -1.25) {$a_n$};
		\node [style=small box] (5) at (1.25, 0) {$g_n$};
		\node [style=none] (6) at (0, 0) {$\cdots$};
	\end{pgfonlayer}
	\begin{pgfonlayer}{edgelayer}
		\draw [style=swap] (3) to (1);
		\draw [style=swap] (5) to (4);
		\draw [style=swap] (0) to (3);
		\draw [style=swap] (2) to (5);
	\end{pgfonlayer}
\end{tikzpicture}
	\]
	then we actually had the following equality between processes in the first place:
	\[
		\begin{tikzpicture}
	\begin{pgfonlayer}{nodelayer}
		\node [style=none] (0) at (-1.25, 1) {};
		\node [style=none] (1) at (-1.25, -1) {};
		\node [style=none] (2) at (1.25, 1) {};
		\node [style=small box] (3) at (-1.25, 0) {$f_1$};
		\node [style=none] (4) at (1.25, -1) {};
		\node [style=small box] (5) at (1.25, 0) {$f_n$};
		\node [style=none] (6) at (0, 0) {$\cdots$};
	\end{pgfonlayer}
	\begin{pgfonlayer}{edgelayer}
		\draw [style=swap] (3) to (1);
		\draw [style=swap] (5) to (4);
		\draw [style=swap] (0) to (3);
		\draw [style=swap] (2) to (5);
	\end{pgfonlayer}
\end{tikzpicture}\ \ = \ \ \begin{tikzpicture}
	\begin{pgfonlayer}{nodelayer}
		\node [style=none] (0) at (-1.25, 1) {};
		\node [style=none] (1) at (-1.25, -1) {};
		\node [style=none] (2) at (1.25, 1) {};
		\node [style=small box] (3) at (-1.25, 0) {$g_1$};
		\node [style=none] (4) at (1.25, -1) {};
		\node [style=small box] (5) at (1.25, 0) {$g_n$};
		\node [style=none] (6) at (0, 0) {$\cdots$};
	\end{pgfonlayer}
	\begin{pgfonlayer}{edgelayer}
		\draw [style=swap] (3) to (1);
		\draw [style=swap] (5) to (4);
		\draw [style=swap] (0) to (3);
		\draw [style=swap] (2) to (5);
	\end{pgfonlayer}
\end{tikzpicture}
	\]
\end{definition}
\noindent Note that the notion of product tomography defined above is much weaker than the notion of \emph{local tomography} \cite{araki1980characterization,bergia1980actual} appearing in a number of reconstructions of quantum theory \cite{HardyAxiom,Chiri2,selby2018reconstructing} and from which the name ``product tomography'' is inspired. For example, it is enough (but by no means necessary) to assume that: 
\begin{enumerate}
	\item[(i)] every process $f\neq 0$ admits some state $a$ and some effect $b$ such that  $b \circ f \circ a$ is an invertible scalar;
	\item[(ii)] if any $f,g:A \rightarrow B$ satisfy $b \circ f \circ a = b \circ g \circ a$ for every state $a$ on $A$ and effect $b$ on $B$, then $f=g$.
\end{enumerate}
\begin{definition}\em
	Let $\mathcal{C}$ be a SMC. We say that a sub-category $\mathcal{A}$ is \emph{$\otimes$-free} if the following conditions hold:
	\begin{enumerate}
		\item[(i)] the tensor unit is an object of $\mathcal{A}$, and all other objects of $\mathcal{A}$ are $\otimes$-prime;
		\item[(ii)] the objects of $\langle \mathcal{A} \rangle_{\otimes}$ are all uniquely $\otimes$-factorisable;\footnote{By condition (i) and definition of the minimal span, the unique factorisation is necessarily in terms of objects of $\mathcal{A}$.}
		\item[(iii)] the SMC $\langle \mathcal{A} \rangle_{\otimes}$ is product tomographic.
	\end{enumerate}
	We say that a SMC $\mathcal{C}$ is \emph{freely interacting} if it has a reflective sub-category in the form $\langle \mathcal{A} \rangle_{\otimes}$, for some $\otimes$-free sub-category $\mathcal{A}$, such that the inclusion-retraction pair is an adjoint equivalence $\mathcal{C} \simeq \langle \mathcal{A} \rangle_{\otimes}$.
\end{definition}

\noindent The reason for the product tomography requirement is that, in its absence, the tensor product itself could be hiding some form of interaction between systems which might not be discoverable by only considering the $\otimes$-free fragment. The following result characterises the categorical correspondence between the non-monoidal ``$\otimes$-free'' perspective and the monoidal ``freely interacting'' perspective.
\begin{proposition}  
	Let $\mathcal{C} \simeq \langle \mathcal{A} \rangle_{\otimes}$ and $\mathcal{D} \simeq \langle \mathcal{B} \rangle_{\boxtimes}$ be two freely interacting SMCs.
	Any functor $F: \mathcal{A} \rightarrow \mathcal{B}$ between the corresponding $\otimes$-free sub-categories which is full on states and effects lifts to an essentially unique\footnote{I.e. unique up to natural monoidal isomorphism.} monoidal functor $\bar{F}:\mathcal{C} \rightarrow \mathcal{D}$.
\end{proposition} 
\begin{proof}  
	We begin by defining a monoidal functor $\hat{F}: \langle \mathcal{A} \rangle_{\otimes} \rightarrow \langle \mathcal{B} \rangle_{\boxtimes}$ as follows, using the fact that the objects of $\langle \mathcal{A} \rangle_{\otimes}$ are uniquely $\otimes$-factorisable 
	\begin{itemize}
		\item on objects, we set $\hat{F}(\otimes_{j=1}^n A_j) := \boxtimes_{j=1}^n F(A_j)$;
		\item on morphisms, we set $\hat{F}(\otimes_{j=1}^n f_j) := \boxtimes_{j=1}^n F(f_j)$;
		\item we respect all associators/unitors/symmetry isomorphisms; 
	\end{itemize}
	The functor $\hat{F}$ will evidently be monoidal, but first we need to check that it is actually well-defined. 

	On objects, well-definition of $\hat{F}$ follows from unique $\otimes$-factorisability. On morphisms, we can restrict our attention to the case of $\hat{F}(\otimes_{j=1}^n f_j)$: all other morphisms can be obtained by associators, unitors and symmetry isomorphisms, which are respected by $\hat{F}$. For every pair of families $(f_j,g_j:A_j \rightarrow B_j)_{j=1}^n$, we need to show that $\otimes_{j=1}^n f_j = \otimes_{j=1}^n g_j$ implies $\hat{F}(\otimes_{j=1}^n f_j) = \hat{F}(\otimes_{j=1}^n g_j)$. If $(a_j:\tensorUnit \rightarrow A_j)_{j=1}^n$ and $(b_j:B_j \rightarrow \tensorUnit)_{j=1}^n$ are arbitrary families of states/effects, then $\otimes_{j=1}^n f_j = \otimes_{j=1}^n g_j$ implies the following:
	\[
		(\otimes_{j=1}^n b_j)\circ(\otimes_{j=1}^n f_j)\circ(\otimes_{j=1}^n a_j)
		=
		(\otimes_{j=1}^n b_j)\circ(\otimes_{j=1}^n g_j)\circ(\otimes_{j=1}^n a_j)
	\]
	Using the exchange law, we can re-write the above as $\otimes_{j=1}^n (b_j \circ f_j\circ a_j)=\otimes_{j=1}^n (b_j \circ g_j\circ a_j)$, and the LHS/RHS get sent to the following by $\Phi$:
	\[
		\begin{array}{rcccl}
			\hat{F}\Big(\otimes_{j=1}^n (b_j \circ f_j\circ a_j)\Big)
			&=&
			\boxtimes_{j=1}^n F(b_j \circ f_j\circ a_j) 
			&=& 
			F(\boxtimes_{j=1}^n b_j) \circ F(\boxtimes_{j=1}^n f_j) \circ F(\boxtimes_{j=1}^n a_j)
			\\
			\hat{F}\Big(\otimes_{j=1}^n (b_j \circ g_j\circ a_j)\Big)
			&=&
			\boxtimes_{j=1}^n F(b_j \circ g_j\circ a_j) 
			&=& 
			F(\boxtimes_{j=1}^n b_j) \circ F(\boxtimes_{j=1}^n g_j) \circ F(\boxtimes_{j=1}^n a_j)
			\\
		\end{array}
	\]
	We can now use product tomography of $\langle \mathcal{B} \rangle_{\boxtimes}$, together with the fact that $F$ is full on states and effects, to conclude that $\hat{F}\Big(\otimes_{j=1}^n f_j\Big)= \boxtimes_{j=1}^n F(f_j)$ and $\hat{F}\Big(\otimes_{j=1}^n g_j\Big)= \boxtimes_{j=1}^n F(g_j)$ are actually the same morphism.

	Having successfully lifted $F: \mathcal{A} \rightarrow \mathcal{B}$ to $\hat{F}:\langle \mathcal{A} \rangle_{\otimes} \rightarrow \langle \mathcal{B} \rangle_{\boxtimes}$, we now obtain a lifting to the freely interacting categories by considering the monoidal functor $\bar{F} := \hat{F}\circ R:\mathcal{C} \rightarrow \mathcal{D}$, where $R: \mathcal{C} \rightarrow \langle \mathcal{A} \rangle_{\otimes}$ is the retraction for the reflective sub-category equivalence.

	Finally, essential uniqueness can be proven as follows. By very construction of the category $\langle \mathcal{A} \rangle_{\otimes}$, the lifting $\hat{F}$ is necessarily unique, so any monoidal $G:\mathcal{C} \rightarrow \mathcal{D}$ which restricts to $G|_{\mathcal{A}} = F: \mathcal{A} \rightarrow \mathcal{B}$ must also restrict to $G|_{\langle \mathcal{A} \rangle_{\otimes}} = \hat{F}: \langle \mathcal{A} \rangle_{\otimes} \rightarrow \langle \mathcal{B} \rangle_{\boxtimes}$. If $E:  \langle \mathcal{A} \rangle_{\otimes} \rightarrow \mathcal{C}$ is the injection for the reflective sub-category equivalence and $\epsilon: E\circ R \stackrel{\isom}{\rightarrow} \id{\mathcal{C}}$ is the co-unit for the equivalence, then we can construct a natural isomorphism $G \epsilon: G \stackrel{\isom}{\rightarrow} G\circ E \circ R$, and we conclude by observing that $G \circ E \circ R = G|_{\langle \mathcal{A} \rangle_{\otimes}} \circ R = \hat{F} \circ R = \bar{F}$.
\end{proof}

\subsection{SMCs with Local Character}  

Having \emph{local character} for a SMC means that data specified on some specific $\otimes$-free subcategory can always be lifted---in an essentially unique way---to the whole SMC. In practice, the existence of such a lifting may require the data to live in a sufficiently structured category, while its uniqueness may require the transformations allowed on the data to be sufficiently rigid. Existence and uniqueness may also depend on the amount of structure possessed by the $\otimes$-free subcategory. 
As a consequence, our notion of local character will be defined relative to two `universes', one specifying the structural constraints for the parent SMC and another one specifying the structural constraints for the $\otimes$-free sub-category.

\begin{definition}\em  
	Let $\CatUniverse$ be the category of (suitably small) categories and functors between them, and let $\SymMonCatUniverse \hookrightarrow \CatUniverse$ be the sub-category of symmetric monoidal categories and monoidal functors between them. We define a \emph{universe} to be a sub-category of $\CatUniverse$, and a \emph{SMC-universe} to be a sub-category of $\SymMonCatUniverse$.
\end{definition}
Specifying a (SMC-)universe is an extremely abstract way of enforcing categorical requirements on theories of interest in a given context. There are many SMC-universes that star in recurring roles in the categorical study of quantum theory and linguistics:
\begin{itemize}
	\item the SMC-universe of SMCs and monoidal functors;
	\item the SMC-universe of compact closed SMCs and monoidal functors;
	\item the SMC-universe of $\CMonCategory$-enriched SMCs and linear monoidal functors;
	\item the SMC-universe of categories of relations over quantales and continuous linear functors;
	\item the SMC-universe of SMCs enriched in $R$-modules and $R$-linear functors between them;
\end{itemize}
Our notion of \emph{local character} will be specified with respect to two such universes: a SMC-universe $\Theta$ for the interacting theory, and a larger universe $\Xi$ for the atomic sub-theory. This means that theories which have local character in the presence of certain structure may not have local character when different structure is chosen instead. For example, we will see later on that $\fHilbCategory$ has local character in the presence of linear structure (when seen as the category $\RMatCategory{\complexs}$), but it's easy to see that it does not in general.
\begin{definition}\em  
	Let $\Xi$ be a universe and $\Theta$ be a SMC-universe. We say that $(\Xi,\Theta)$ is a \emph{$\otimes$-free/interacting pair of universes} if the following conditions hold:
	\begin{enumerate}
		\item $\Theta$ is a sub-category of $\Xi$;
		\item if $\mathcal{C}$ is an SMC in $\obj{\Theta}$ and $\mathcal{A}$ is a $\otimes$-free sub-category of $\mathcal{C}$, then $\mathcal{A}$ is a category in $\Xi$ and the sub-category inclusion $\mathcal{A} \hookrightarrow \mathcal{C}$ is a functor in $\Xi$.
	\end{enumerate}
\end{definition}
\begin{definition}\em  
	Let $(\Xi,\Theta)$ be a $\otimes$-free/interacting pair of universes. Let $\mathcal{C} \in \obj{\Theta}$ be a SMC, and let $\mathcal{A}$ be a $\otimes$-free sub-category of $\mathcal{C}$ which satisfies the following \emph{lifting property}:
	\begin{itemize}
		\item for every $\mathcal{D} \in \obj{\Theta}$ and every functor $F: \mathcal{A} \rightarrow \mathcal{D}$ in $\Xi$ which sends the tensor unit of $\mathcal{C}$ to the tensor unit of $\mathcal{D}$, there is an essentially unique\footnote{I.e. unique in $\Theta$ up to natural monoidal isomorphism.} monoidal functor $\bar{F} : \mathcal{C} \rightarrow \mathcal{D}$ in $\Theta$ with $\bar{F}|_{\mathcal{A}} = F$. 
	\end{itemize}
	We say that $\mathcal{C}$ has \emph{local character given by $\mathcal{A}$} with respect to $(\Xi,\Theta)$.
\end{definition}
Intuitively, we could think of the universes $\Xi$ and $\Theta$ as the domain and codomain of a free construction $F:\mathcal{A} \mapsto \mathcal{C}$: the lifting property would act as some kind of weak universal property establishing a weak adjunction $\Hom{\Xi}{\mathcal{A}}{\mathcal{D}} \simeq \Hom{\Theta}{\mathcal{C}}{\mathcal{D}}$ between $F:\Xi \rightarrow \Theta$ and the inclusion functor $\Theta \hookrightarrow \Xi$. This intuition provides good guidance when looking at the uniqueness result below, but one should be careful not to take the analogy too literally: the free construction above is not well-defined at all. There are a number of good reasons for this, some of which are listed below, but ultimately the issue boils down to the fact that the interesting object of study is the SMC $\mathcal{C}$, and not the sub-category $\mathcal{A}$.
\begin{itemize}
	\item Not all choices of $\otimes$-free sub-category $\mathcal{A}$ of $\mathcal{C}$ in $\Xi$ are guaranteed to work, in the sense that the lifting property will be satisfied, and the working choices are not guaranteed to be unique or natural.
	\item The category $\mathcal{C}$ is not determined by the category $\mathcal{A}$ in a unique or natural way, so the functor $F:\mathcal{A} \mapsto \mathcal{C}$ is not well-defined.
	\item The categories $\mathcal{C}$ and $\mathcal{A}$ have constraints not satisfied by the category $\mathcal{D}$, so the homsets for the adjunction are not well-defined.
\end{itemize}
In fact, not even the (very special) freely interacting case $\mathcal{A} \mapsto \langle \mathcal{A} \rangle_{\otimes}$ is well-defined: the SMC $\langle \mathcal{A} \rangle_{\otimes}$ is not, in general, the free SMC on $\mathcal{A}$, depending instead on the specific tensor structure of the parent $\mathcal{C}$.

Having clarified this, we are now in a position to formulate our uniqueness result. Intuitively, we wish to show that there is at most one way---up to equivalence---of turning a specified $\otimes$-free theory into an interacting theory with local character (w.r.t. a fixed $\otimes$-free/interacting pair of universes). In other words, we wish to show that \emph{having local character given by some $\otimes$-free theory is a property} for SMCs, at least with respect to a specified $\otimes$-free/interacting pair of universes. Because the universes might impose arbitrary requirements on their theories, the result is more clearly formulated from the outside-in: we start from two theories with local character given by the same $\otimes$-free sub-theory (up to isomorphism) and we show that they must be equivalent in the chosen SMC-universe.
\begin{theorem}
\label{thm_main}
	Let $(\Xi,\Theta)$ be a $\otimes$-free/interacting pair of universes. Let $\mathcal{C}, \mathcal{D} \in \obj{\Theta}$ have local character w.r.t. $(\Xi,\Theta)$, and assume that the $\otimes$-free sub-categories $\mathcal{A},\mathcal{B}$ giving them local character are isomorphic in $\Xi$. Then there exists a monoidal equivalence of categories $\mathcal{C} \simeq \mathcal{D}$ in $\Theta$. Furthermore, the monoidal equivalence restricts to the chosen isomorphisms $\xi:\mathcal{A}\stackrel{\isom}{\rightarrow}\mathcal{B}$ and $\xi^{-1}:\mathcal{B}\stackrel{\isom}{\rightarrow}\mathcal{A}$, and it is essentially the only one in $\Theta$ doing so.
\end{theorem}
\begin{proof}
	Consider the isomorphism $\xi: \mathcal{A} \stackrel{\isom}{\rightarrow} \mathcal{B}$, $\xi^{-1}: \mathcal{B} \stackrel{\isom}{\rightarrow} \mathcal{A}$, the sub-category injection $E_{\mathcal{A}}: \mathcal{A} \hookrightarrow \mathcal{C}$ and the sub-category injection $E_{\mathcal{B}}: \mathcal{B} \hookrightarrow \mathcal{D}$; these are all functors in $\Xi$. By local character of $\mathcal{C}$, the functor $\Phi:= E_{\mathcal{B}} \circ \xi: \mathcal{A} \rightarrow \mathcal{D}$ has an essentially unique lifting $\bar{\Phi}: \mathcal{C} \rightarrow \mathcal{D}$ in $\Theta$. By local character of $\mathcal{D}$, the functor $\Psi:= E_{\mathcal{A}} \circ \xi^{-1}: \mathcal{B} \rightarrow \mathcal{C}$ has an essentially unique lifting $\bar{\Psi}: \mathcal{D} \rightarrow \mathcal{C}$ in $\Theta$.

	We have used the existence clause of local character to construct $\bar{\Phi}$ and $\bar{\Psi}$, and now we will use the essential uniqueness clause to turn them into an equivalence. Consider the functor $\bar{\Psi} \circ \bar{\Phi}: \CategoryC \rightarrow \CategoryC$ in $\Theta$. Because of the way $\Phi$ and $\Psi$ were defined, when restricted to $\mathcal{A}$ this functor gives the identity:
	\[
		\big(\bar{\Psi} \circ \bar{\Phi}\big)|_{\mathcal{A}} = \xi^{-1}\circ\xi = \id{\mathcal{A}} = \big(\id{\mathcal{C}}\big)|_{\mathcal{A}}
	\] 
	The essential uniqueness clause of local character for $\mathcal{C}$ gives a natural isomorphism $\bar{\Psi} \circ \bar{\Phi} \isom \id{\mathcal{C}}$. A symmetric argument then gives a natural isomorphism $\id{\mathcal{D}} \isom \bar{\Phi} \circ \bar{\Psi}$, showing that there is an adjoint equivalence in $\Theta$ between $\mathcal{C}$ and $\mathcal{D}$. Finally, consider another adjoint equivalence $F: \mathcal{C} \rightarrow \mathcal{D}$ and $G: \mathcal{D} \rightarrow \mathcal{C}$ such that $F|_{\mathcal{X}} = \xi$ and $G|_{\mathcal{B}} = \xi^{-1}$: by the essential uniqueness clause of local character, we immediately conclude that $F \isom \bar{\Phi}$ and $G \isom \bar{\Psi}$. 
\end{proof}
In principle, a theory $\mathcal{C}$ with local character can have many more systems than those freely generated by the $\otimes$-free sub-theory. Using our uniqueness result, however, we can show that the additional systems don't add anything essential to $\mathcal{C}$ (save from satisfying any structural requirements imposed by the choice of SMC-universe). 
\begin{proposition}  
	Let $(\Xi,\Theta)$ be a $\otimes$-free/interacting pair of universes. Let $\mathcal{C}$ be a SMC in $\Theta$ which has local character w.r.t. $(\Xi,\Theta)$ given by $\otimes$-free sub-category $\mathcal{A}$. Assume that $\mathcal{C}'$ is a reflective sub-SMC of $\mathcal{C}$ in $\Theta$
		\footnote{We intend the sub-SMC $\mathcal{C}'$, the inclusion $E_{\mathcal{C}'}:\mathcal{C}'\hookrightarrow \mathcal{C}$ and the retraction $R_{\mathcal{C}'}:\mathcal{C}\rightarrow \mathcal{C}'$ all to be in $\Theta$.} 
	such that $\langle\mathcal{A}\rangle_{\otimes} \hookrightarrow \mathcal{C}'$ in $\Theta$.
	Then $\mathcal{C}'$ also has local character w.r.t. $(\Xi,\Theta)$ given by $\mathcal{A}$, and the injection-retraction pair gives an adjoint monoidal equivalence of categories $\mathcal{C} \simeq \mathcal{C}'$ in $\Theta$.
\end{proposition}
\begin{proof}
	We consider the inclusion $E_{\langle\mathcal{A}\rangle_{\otimes}}:\langle\mathcal{A}\rangle_{\otimes} \hookrightarrow \mathcal{C}$,
	the inclusion $E_{\langle\mathcal{A}\rangle_{\otimes},\mathcal{C}'}:\langle\mathcal{A}\rangle_{\otimes} \hookrightarrow \mathcal{C}'$, 
	the inclusion $E_{\mathcal{C}'}:\mathcal{C}'\hookrightarrow \mathcal{C}$ 
	and the retraction $R_{\mathcal{C}'}:\mathcal{C}\rightarrow \mathcal{C}'$; these are all functors in $\Theta$. We begin by showing that $\mathcal{C}'$ inherits local character from $\mathcal{C}$. 

	Given another SMC $\mathcal{D}$ in $\Theta$ and a monoidal functor $F: \mathcal{A} \rightarrow \mathcal{D}$ in $\Xi$, we consider the lifting $\bar{F}:\mathcal{C} \rightarrow \mathcal{D}$ given by local character of $\mathcal{C}$ and we construct the functor $\bar{F}\circ E_{\mathcal{C}'}: \mathcal{C}' \rightarrow \mathcal{D}$ in $\Theta$. The functor $\bar{F}\circ E_{\mathcal{C}'}$ is a lifting of $F$ because: 
	\[
		\big(\bar{F}\circ E_{\mathcal{C}'}\big)|_{\mathcal{A}} 
		= 
		\big(\bar{F}\circ E_{\mathcal{C}'}\circ E_{\langle\mathcal{A}\rangle_{\otimes},\mathcal{C}'}\big)|_{\mathcal{A}} 
		=
		\big(\bar{F}\circ E_{\langle\mathcal{A}\rangle_{\otimes}}\big)|_{\mathcal{A}}
		=
		\bar{F}|_{\mathcal{A}}
		=
		F
	\]
	This proves existence of a lifting from $F: \mathcal{A} \rightarrow \mathcal{D}$ to $\bar{F}\circ E_{\mathcal{C}'}: \mathcal{C}' \rightarrow \mathcal{D}$.

	To prove essential uniqueness of the lifting $\bar{F}\circ E_{\mathcal{C}'}$, let $G: \mathcal{C}' \rightarrow \mathcal{D}$ be another functor in $\Theta$ such that $G|_{\mathcal{A}} = F$, and consider the functor $G \circ R_{\mathcal{C}'} : \mathcal{C} \rightarrow \mathcal{D}$ in $\Theta$. This is a lifting of $F:  \mathcal{A} \rightarrow \mathcal{D}$ because:
	\[
		\big(G \circ R_{\mathcal{C}'}\big)|_{\mathcal{A}} 
		= 
		\big(G\circ R_{\mathcal{C}'} \circ E_{\mathcal{C}'} \big)|_{\mathcal{A}} 
		=
		G|_{\mathcal{A}}
		=
		F
	\]
	By the essential uniqueness clause for local character of $\mathcal{C}$, we get that there is a natural isomorphism $\varphi: G \circ R_{\mathcal{C}'} \stackrel{\isom}{\Rightarrow} \bar{F}$. This in turn gives a natural isomorphism $\varphi_{E_{\mathcal{C}'}} : G \circ R_{\mathcal{C}'} \circ E_{\mathcal{C}'} \stackrel{\isom}{\Rightarrow} \bar{F} \circ E_{\mathcal{C}'}$, and we conclude by observing that $G \circ R_{\mathcal{C}'} \circ E_{\mathcal{C}'} = G$.

	Having established that $\mathcal{A}$ gives local character to the sub-SMC $\mathcal{C}'$, we with to use Theorem \ref{thm_main} to show that $E_{\mathcal{C}'}:\mathcal{C}'\rightarrow\mathcal{C}$ and $R_{\mathcal{C}'}:\mathcal{C}\rightarrow \mathcal{C}'$ form an adjoint monoidal equivalence of categories $\mathcal{C} \simeq \mathcal{C}'$. We appeal to Theorem \ref{thm_main} to $\mathcal{C}$ and $\mathcal{D}:=\mathcal{C}'$, using $\xi := \id{\mathcal{A}}: \mathcal{A} \rightarrow \mathcal{A}$ as our chosen isomorphism, and we obtain an adjoint monoidal equivalence of categories $F: \mathcal{C}' \rightarrow \mathcal{C}$ and $G: \mathcal{C} \rightarrow \mathcal{C}'$ such that $F|_{\mathcal{A}} = \id{\mathcal{A}} = G|_{\mathcal{A}}$. By the essential uniqueness clauses for local character of $\mathcal{C}$ and $\mathcal{C}'$ respectively, we conclude that there are natural monoidal isomorphisms $\varphi: E_{\mathcal{C}'} \stackrel{\isom}\Rightarrow F$ and $\psi: R_{\mathcal{C}'} \stackrel{\isom}\Rightarrow G$. By adjoint equivalence we also get natural monoidal isomorphisms $\epsilon: F \circ G\stackrel{\isom}{\Rightarrow}\id{\mathcal{C}}$ and $\eta: \id{\mathcal{C}'}\stackrel{\isom}{\Rightarrow}G \circ F$. We can compose these natural monoidal isomorphisms horizontally (denoted by $\ast$) and vertically (denoted by $\cdot$) to obtain natural monoidal isomorphisms showing that $E_{\mathcal{C}'}$ and $R_{\mathcal{C}'}$ form an adjoint monoidal equivalence with co-unit $\epsilon \cdot (\psi \ast \varphi): E\circ R \stackrel{\isom}{\Rightarrow} \id{\mathcal{C}}$ and unit $(\varphi^{-1} \ast \psi^{-1})\cdot \eta: \id{\mathcal{C}'} \stackrel{\isom}{\Rightarrow} R \circ E$.
\end{proof}

%

\section{Examples}\label{section_applications} 

We now show local character for three large families of symmetric monoidal categories of interest: 
\begin{itemize}
	\item free finite-dimensional modules over a commutative semiring, showing that our new result generalises the result originally presented in \cite{Uniqueness2017};
	\item relations over a quantale, showing that our new result applies to infinite-dimensional examples;
	\item finitely generated modules over a principal ideal domain, showing that our new result applies to non-free examples;
\end{itemize}
Relations over quantales are an important class of examples: they are fundamental in the monoidal approach to topology~\cite{HofmannSealTholen2014} and they have recently found application in compositional models of language and cognition~\cite{CoeckeGenoveseLewisMarsdenToumi2018}. Intuitively the quantale values can be seen to describe quantities such as connection strengths, costs, distances and success probabilities, following ideas originally due to Lawvere~\cite{Lawvere1973}. Finitely generated $R$-modules are another important class of examples, connecting our result to the historic uniqueness results by Eilenberg and Watts \cite{Eilenberg1960,Watts1960}.
\begin{theorem}  
	\label{thm_SMat}
	Let $S$ be a commutative semiring, let $\Xi$ be the universe of categories enriched in $S$-modules and $S$-linear functors between them. Let $\Theta$ be the SMC-universe of SMCs enriched in $S$-modules and $S$-linear monoidal functors between them. The SMC $\RMatCategory{S}$ of free finite-dimensional modules over $S$ has local character w.r.t. $(\Xi,\Theta)$.
\end{theorem}
\begin{proof}
	We define the $\otimes$-free sub-category $\mathcal{A}$ to be the full sub-category with objects in the form $A(p) := S^p = \oplus_{d=0}^{p-1} S$, where $p$ is a prime number. Because morphisms $A(p) \rightarrow A(q)$ are all the $S$-linear maps $S^p \rightarrow S^q$, the category $\mathcal{A}$ clearly lives in the $\otimes$-free universe $\Xi$. We write $\ket{a_0^{(p)}},...,\ket{a_{p-1}^{(p)}}$ for the standard orthonormal basis of states for $A(p)$ and $\bra{a_0^{(p)}},...,\bra{a_{p-1}^{(p)}}$ for the corresponding effects (such that $\braket{a_i^{(p)}}{a_j^{(p)}} = \delta_{ij}$). 

	The generic object of the freely interacting sub-SMC $\langle \mathcal{A} \rangle_{\otimes}$ takes the form $\otimes_{i=1}^n A(p_i)$---up to associators and unitors---and morphisms $\otimes_{i=1}^n A(p_i) \rightarrow \otimes_{j=1}^m A(q_i)$ are certain $S$-linear combinations of the following atomic morphisms:
	\[
		\Big(\bigotimes_{j=1}^m \ket{a_{h_j}^{(q_j)}}\Big) \circ \Big(\bigotimes_{i=1}^n \bra{a_{k_i}^{(p_i)}}\Big) 
	\]
	As a consequence, it is easy to check that the freely-interacting sub-SMC $\langle \mathcal{A} \rangle_{\otimes}$ is product tomographic, as required by the definition of a $\otimes$-free sub-category.

	We now consider the maximal span $\overline{\langle\mathcal{A}\rangle_{\otimes}}$ of $\mathcal{A}$ in $\RMatCategory{S}$, i.e. the full sub-SMC generated by the objects of $\langle\mathcal{A}\rangle_{\otimes}$. A generic object of $\RMatCategory{S}$ is a finite set $X$, and we consider the prime factorisation $\# X= \prod_{i=1}^{n_x} p_i^{(X)}$ of its cardinality---where some of the factors $p_i$ may be equal, and we write $1$ for the empty product---and we get an $S$-linear isomorphism $\eta_X: X \isom \otimes_{i=1}^{n_X} A(p_i^{(X)})$. Starting from the sets with prime cardinality, it is always possible to choose these isomorphisms in such a way that $\eta_{X\otimes Y} = \eta_{X} \otimes \eta_{Y}$. This can be used to define the following retraction $R: \RMatCategory{S} \rightarrow \overline{\langle\mathcal{A}\rangle_{\otimes}}$ for the sub-category injection $E:\overline{\langle\mathcal{A}\rangle_{\otimes}} \hookrightarrow \RMatCategory{S}$:
	\[
		\begin{array}{rcl}
			R(X) & := & \otimes_{i=1}^{n_X} A(p_i^{(X)})\\
			R(f:X \rightarrow Y) & := & \epsilon_Y \circ f \circ \epsilon_X^{-1}
		\end{array}
	\]
	The retraction $R$ is monoidal and $S$-linear, so it turns $\overline{\langle\mathcal{A}\rangle_{\otimes}}$ into a reflective sub-SMC of $\RMatCategory{S}$ in $\Theta$. Furthermore, the injection-retraction pair is an $S$-linear adjoint monoidal equivalence $\RMatCategory{S} \simeq \overline{\langle\mathcal{A}\rangle_{\otimes}}$ in $\Theta$. 

	We can therefore restrict our attention to the maximal span $\overline{\langle\mathcal{A}\rangle_{\otimes}}$, where the generic morphism $\otimes_{i=1}^n A(p_i) \rightarrow \otimes_{j=1}^m A(q_j)$ is a generic $S$-valued matrix:
	\[
		M := \sum_{k_1=0}^{p_1-1} ... \sum_{k_n=0}^{p_n-1}\sum_{h_1=0}^{q_1-1} ... \sum_{h_m=0}^{q_m-1} M_{k_1 ... k_n h_1 ... h_n} \Big(\otimes_{j=1}^m \ket{a_{h_j}^{(q_j)}}\Big) \circ \Big(\otimes_{i=1}^n \bra{a_{k_i}^{(p_i)}}\Big) 
	\]
	Given an SMC $\mathcal{D}$ enriched in $S$-modules (i.e. one in $\Theta$) and an $S$-linear functor $F: \mathcal{A} \rightarrow \mathcal{D}$ (i.e. one in $\Xi$), a lifting $\hat{F}:\overline{\langle\mathcal{A}\rangle_{\otimes}}\rightarrow \mathcal{D}$ can be defined as follows:
	\[
		\hat{F}(M) := \sum_{k_1=0}^{p_1-1} ... \sum_{k_n=0}^{p_n-1}\sum_{h_1=0}^{q_1-1} ... \sum_{h_m=0}^{q_m-1} M_{k_1 ... k_n h_1 ... h_n} \Big(\boxtimes_{j=1}^m F\big[\ket{a_{h_j}^{(q_j)}}\big]\Big) \circ \Big(\boxtimes_{i=1}^n F\big[\bra{a_{k_i}^{(p_i)}}\big]\Big)
	\]
	It is easy to check that the functor $\hat{F}$ is well-defined and restricts to $F$ on $\mathcal{A}$. The functor $\hat{F}$ is also $S$-linear and monoidal (i.e. in $\Theta$), so we can extend it to an $S$-linear monoidal functor $\bar{F}:=\hat{F} \circ R:\RMatCategory{S}\rightarrow \mathcal{D}$ to proving existence of a lifting $F \mapsto \bar{F}$. 

	The proof of essential uniqueness of the lifting $\bar{F}$ goes as follows. If $G: \RMatCategory{S}\rightarrow \mathcal{D}$ is an $S$-linear monoidal functor such that $G|_{\mathcal{A}} = F$, then by $S$-linearity we necessarily have that $G \circ E = \hat{F}$, from which it follows that $G \circ E \circ R = \hat{F}\circ R = \bar{F}$. From the monoidal natural isomorphism $\eta: \id{\RMatCategory{S}} \stackrel{\isom}{\Rightarrow} E \circ R$ we finally get the desired monoidal natural isomorphism $G \eta: G \stackrel{\isom}{\Rightarrow} G \circ E \circ R = \bar{F}$.
\end{proof}
The proof of local character for $\RMatCategory{S}$ doesn't make any explicit use of dimensional rigidity, a key ingredient of the original proof that prevented its extension to infinite-dimensional and non-free settings. With some tweaking, we can now extend our proof to categories of relations over quantales (which are infinite-dimensional) and to certain categories of modules over semirings (which are non-free). Theorem \ref{thm_main} can then be invoked to conclude that those categories have an essentially unique symmetric monoidal structure with local character within their relevant SMC-universe.
\begin{theorem}  
	\label{thm_RelQ}
	Let $Q$ be a quantale, let $\Xi$ be the universe of categories enriched in $Q$-modules
		\footnote{By which we mean complete join semilattices with an action of $Q$.}
	and continuous $Q$-linear functors between them. 
	Let $\Theta$ be the SMC-universe of SMCs enriched in $Q$-modules and continuous $Q$-linear monoidal functors between them. The SMC $\RelQCategory{Q}$ of $Q$-valued relations\footnote{Restricted to sets smaller than some suitably large infinite cardinal $\aleph_{\beta}$.} has local character w.r.t. $(\Xi,\Theta)$.
\end{theorem}
\begin{proof}
	The proof is the same as for the previous result, save for the following variations: 
	\begin{enumerate}
		\item[(i)] instead of restricting our attention to finite prime ordinals $\{0,...,p-1\}$, we also include all infinite initial ordinals $\omega_\alpha$;\footnote{Because of our restriction on the cardinality of sets, these are exactly all infinite initial ordinals $\omega_\alpha$ with $\alpha<\beta$.}
		\item[(ii)] instead of using $\sum\limits_{k=0}^{p-1}$, we use $\bigvee\limits_{k = 0}^{p-1}$ for finite ordinals and $\bigvee\limits_{k < \omega_{\alpha}}$ for infinite initial ordinals.
	\end{enumerate}
\end{proof}
Save for the change from finitary to infinitary operations, the proof for the category of relations over a quantale still takes place in a free setting, where morphisms are matrices. The move to a non-free setting instead requires some additional sophistication, so the proof below is presented in full detail.
\begin{theorem}  
	Let $R$ be a principal ideal domain, let $\Xi$ be the universe of categories enriched in $R$-modules and $R$-linear functors between them. Let $\Theta$ be the SMC-universe of SMCs enriched in $R$-modules and $R$-linear monoidal functors between them. The SMC $\RModCategory{R}_{fg}$ of finitely generated $R$-modules has local character w.r.t. $(\Xi,\Theta)$.
\end{theorem}
\begin{proof}
	The proof is conceptually the same given above for the free finite-dimensional $R$-modules, but there are a number of technical variations that need to be carefully spelled out.

	Our modules are no longer free, so the standard orthonormal basis cannot be used to express morphisms as matrices any longer. Instead, we invoke the structure theorem for finitely generated modules over a principal ideal domain to decompose any object $M$ of $\RModCategory{R}_{fg}$ as a finite direct sum of cyclic modules in the following form, where $(r_i^{(M)})$ are all primary ideals (equivalently, $R/(r_i^{(M)})$ are all indecomposable $R$-modules):
		\[
			M \isom \bigoplus_{i=1}^{\dim{M}} R/(r_i^{(M)}) 
		\]
	When objects are decomposed in this form, the tensor product can be written as follows, where the element $\gcd(r_i^{(M)},r_j^{(N)})$ is a generator of the sum ideal $(r_i^{(M)})+(r_j^{(N)})$ in the principal ideal domain $R$:
	\[
		\bigg(\bigoplus_{i=1}^{\dim{M}} R/(r_i^{(M)})\bigg) \otimes \bigg(\bigoplus_{j=1}^{\dim{N}} R/(r_j^{(N)})\bigg) 
		=
		\bigoplus_{i=1}^{\dim{M}} \bigoplus_{j=1}^{\dim{N}} R/\big(\gcd(r_i^{(M)},r_j^{(N)})\big)
	\]

	While the dimension $\dim{M}$ of an object is well-defined, it is no longer enough to take objects with prime dimension, since the dimension of the tensor product is not in general the product of dimensions. For example, for $R=\integers$ we get that the following tensor product of two 2-dimensional objects is 1-dimensional:
	\[
		\begin{array}{rcl}
		\integers/(6) \otimes \integers/(15)
		&=&
		\big(\integers/(2)\oplus\integers/(3)\big) \otimes \big(\integers/(3) \oplus \integers/(5)\big) 
		\\
		&=& 
		\big(\integers/(2) \otimes \integers/(3)\big)
		\oplus
		\big(\integers/(2) \otimes \integers/(5)\big)
		\oplus
		\big(\integers/(3) \otimes \integers/(3)\big)
		\oplus
		\big(\integers/(3) \otimes \integers/(5)\big)
		\\
		&=&
		\integers/(1) 
		\oplus
		\integers/(1) 
		\oplus
		\integers/(3) 
		\oplus
		\integers/(1) 
		\\
		&=&
		\integers/(3)
		\end{array}
	\]
	As the objects of the $\otimes$-free sub-category $\mathcal{A}$ we simply take the $R$-modules in the form $A(r_1,...,r_d):=\bigoplus_{i=1}^{d} R/(r_i)$ which cannot be written as tensor products.


	Now we consider a category $\mathcal{D}$ enriched in $R$-modules and an $R$-linear functor $F: \mathcal{A} \rightarrow \mathcal{D}$. Given two families $A_1,...,A_n$ and $B_1,...,B_n$ of objects of $\mathcal{A}$, we look at the $R$-modules $\mathcal{A}[A_i,B_i]$ and we define a multi-linear map $\prod_{i=1}^n \mathcal{A}[A_i,B_i] \rightarrow \mathcal{D}[\boxtimes_{i=1}^n F(A_i),\boxtimes_{i=1}^n F(B_i)] $ as follows:
	\[
		(f_1,...,f_n) \mapsto F(f_1) \boxtimes ... \boxtimes F(f_n)
	\]
	By the universal property of the tensor product of $R$-modules, this lifts to a unique $R$-module homomorphim $\hat{F}: \otimes_{i=1}^n \mathcal{A}[A_i,B_i] \rightarrow\mathcal{D}[\boxtimes_{i=1}^n F(A_i),\boxtimes_{i=1}^n F(B_i)] $. We observe that the tensor product of two homsets and the homset for the tensor product of domains/codomains coincide in $\RModCategory{R}_{fg}$, so that we have the identification $\otimes_{i=1}^n \mathcal{A}[A_i,B_i] = \RModCategory{R}_{fg}[\otimes_{i-1}^n A_i, \otimes_{i=1}^n B_i]$.

	We can take all these homomorphisms $\RModCategory{R}_{fg}[\otimes_{i-1}^n A_i, \otimes_{i=1}^n B_i] \rightarrow \mathcal{D}[\boxtimes_{i=1}^n F(A_i),\boxtimes_{i=1}^n F(B_i)] $ together to obtain a unique functor $\hat{F}:\overline{\langle\mathcal{A}\rangle_{\otimes}} \rightarrow \mathcal{D}$ which restricts to $F: \mathcal{A} \rightarrow \mathcal{D}$ over $\mathcal{A}$: monoidality of the resulting $\hat{F}$ is clear by construction, and functoriality follows from the fact that all morphisms in $\RModCategory{R}_{fg}[\otimes_{i-1}^n A_i, \otimes_{i=1}^n B_i]$ are $R$-linear combinations of separable ones in the form $f_1 \otimes... \otimes f_n$, over which $\hat{F}$ is automatically functorial.

	Finally, the same reasoning given above for free finite-dimensional $R$-modules can be used to lift $\hat{F}$ to an essentially unique $\bar{F}: \RModCategory{R}_{fg} \rightarrow \mathcal{D}$, completing our proof.
\end{proof}

The three examples presented above are all linear in nature, but this doesn't mean that more classical, non-linear examples are excluded from their reach. Indeed, the Cartesian SMC $\fSetCategory$ of finite sets and functions is a sub-SMC of $\RMatCategory{S}$ for all commutative semirings $S$, and the larger Cartesian SMC $\SetCategory$ of (suitably small) sets and functions is a sub-SMC of $\RelQCategory{Q}$ for all quantales $Q$. 

It is tempting to think that the notion of local character should straightforwardly apply to the Cartesian setting: after all, Cartesian SMCs are seen as modelling minimally interacting theories. However, some care should be taken in defining what exactly should be Cartesian in a theory: minimal interaction is a physical property, so the correct requirement in this context should be for the tensor product to be Cartesian \emph{in the sub-SMC of physical/normalised states and processes}. More freedom can be granted to the parent SMC which contains the building blocks used to understand the physical processes, and this freedom is extremely important from an operational perspective: requiring the whole category to be Cartesian would mean that one does not have enough effects to test properties of systems.

\newcommand{\pFunCategory}{\operatorname{pFun}}
\newcommand{\fpFunCategory}{\operatorname{fpFun}}
By themselves, the categories $\fSetCategory$ and $\SetCategory$ don't have enough effects to allow an operational interpretation as the one advocated above, and they cannot satisfy product tomography. In particular, this means that we cannot expect them to have local character. In all the linear contexts presented above, however, the SMCs $\fSetCategory$ and $\SetCategory$ arise naturally as the normalised sub-SMCs of $\fpFunCategory$ and $\pFunCategory$, the sub-SMCs of (finite) sets and \emph{partial} functions between them, equipped with the environment structure given by the total functions to the singleton.\footnote{Note that the tensor product on $\pFunCategory$ is not the Cartesian one, but the one inherited from the Kronecker product of matrices.} This larger context of partial functions \emph{does} have enough effects to test all properties of sets, and it is the smallest one to do so: as a consequence, it is interesting to ask the question whether $\fpFunCategory$ and $\pFunCategory$---rather than $\fSetCategory$ and $\SetCategory$---have local character. The traditional context of investigation for sets and partial functions is that of categories enriched in pointed DCPOs, so that is the one with consider first by sketching the following tentative result (leaving further investigation to future work).

\begin{conjecture}
	Let $\Xi$ be the universe of categories enriched in pointed DCPOs and Scott-continuous functors between them, respecting finite coproducts and initial objects. Let $\Theta$ be the universe of SMCs enriched in pointed DCPOs and Scott-continuous monoidal functors between them, respecting finite coproducts and initial objects. The sub-category of finite prime ordinals with total functions between them gives $\fpFunCategory$ local character w.r.t. $(\Xi,\Theta)$. The sub-category of finite prime ordinals and infinite initial ordinals with total functions between them gives $\pFunCategory$ local character w.r.t. $(\Xi,\Theta)$.
\end{conjecture}
\begin{proof}(sketch)
	The proof should essentially be analogous to the proofs previously given for the free finite-dimensional case of $\fRelCategory$ and the free infinite-dimensional case of $\RelCategory$, with the DCPO structure and Scott-continuity of functors playing the role that linear structure and linearity of functors played in the original proofs. Furthermore, every partial function is a disjoint union of a total function and a zero function, so the requirement that functors preserve coproducts and initial objects can be used to reduce their definition on partial functions to their definition on total functions.
\end{proof}

The formulation of the above tentative result suggests that pointed DCPO structure might not quite be enough to provide local character to $\fpFunCategory$ and $\pFunCategory$, so we are inspired to look at the problem from a slightly different angle. We note that $\fpFunCategory$ is a sub-SMC of $\RMatCategory{\mathbb{B}}$ for the boolean semiring $\mathbb{B}$, and that $\pFunCategory$ is a sub-SMC of $\RelQCategory{\mathbb{B}}$. The $\mathbb{B}$-module enrichment fails on those sub-categories, but only in the sense that the additive operation $\vee$ fails to be defined on all pairs: as long as we can appropriately deal with partial addition---a challenge in itself, given the existence of non-trivial interactions with function composition and the zero partial function---extensions of Theorems \ref{thm_SMat} and \ref{thm_RelQ} should be rather straightforward. Investigation of how this could be best achieved is also left to future work.

\section{Conclusions and Future Work}
\label{section_conclusions}

We have defined a new notion of symmetric monoidal category with \emph{local character}, based on the intuition that certain interacting process theories are fully described, within an appropriate categorical context, by some $\otimes$-free sub-theory. As our central contributions, we have proven that symmetric monoidal structure with local character is essentially unique when it exists, and that large families of categorical examples of interest are covered by our result. In particular, we managed to include the infinite-dimensional case of categories of relations over quantales---of interest in the monoidal approach to topology and in the study of compositional distributional models of meaning---and the non-free case of finitely-generated modules over principal ideal domains---bringing us closer to a different and well-established uniqueness result by Eilenberg and Watts. We have also sketched a proof that our framework covers the Cartesian setting of sets and (partial) functions, although further investigation of the matter was left to future work.

While the setting presented in this work spans a rather wide spectrum of categorical examples, a number of questions are left open. Firstly, our uniqueness result frames local character as a sufficient condition for uniqueness of symmetric monoidal structure, but does not provide any indication of whether is is also necessary, or how much space might lie between it and a suitable necessary condition. Secondly, the appearance of the universal property for the tensor product in the proof of local character for finite-dimensional $R$-modules over a principal ideal domain suggests that a much more general result proving local character for linear symmetric monoidal categories could be formulated. Finally, the proof of local character for categories of relations over a quantale could likely be extended to categories with non-idempotent infinitary algebraic operations, perhaps from a suitable topos-theoretic perspective. There is also an open question about whether $\otimes$-free subcategories are unique or natural in an appropriate sense, and if so under which conditions.     

\bibliographystyle{eptcs}
\bibliography{biblio}
\nocite{*}

\end{document}